\documentclass[10pt,a4paper]{article}
\usepackage[utf8]{inputenc}
\usepackage[english]{babel}
\usepackage{amsmath}
\usepackage{amsfonts}
\usepackage{amssymb}
\usepackage{dsfont}
\usepackage{amsthm}
\usepackage{float}
\usepackage{textcomp}
\usepackage{faktor}
\usepackage{xcolor}
\usepackage{graphicx}
\usepackage{amssymb}
\usepackage{fancybox}
\usepackage{mathrsfs}
\usepackage{mathtools}
\usepackage{tikz}
\usepackage{centernot}
\usepackage{esint}
\usepackage{setspace}
\usepackage{pgfplots}
\usepackage{multicol}
\usepackage{authblk}
\usepackage{appendix}
\usepackage{fullpage}
\usepackage{enumerate}
\usepackage{algorithm}
\usepackage[noend]{algpseudocode}
\usepackage{hyperref}
\usepackage{comment,verbatim}

\theoremstyle{plain}
\newtheorem{thm}{Theorem}[section]
\newtheorem{lem}[thm]{Lemma}
\newtheorem{prop}[thm]{Proposition}
\newtheorem{cor}[thm]{Corollary}

\theoremstyle{definition}
\newtheorem{defn}[thm]{Definition}

\newtheorem{exmp}[thm]{Example}

\newtheorem{prob}{Problem}
\newtheorem{kprob}{Key Problem}
\newtheorem{cond}{Additional Condition}
\newtheorem{assumption}{Assumption}

\theoremstyle{remark}
\newtheorem{rem}[thm]{Remark}

\newtheorem*{notat}{Notation}

\DeclareMathOperator{\supp} {supp}
\DeclareMathOperator{\w} {w}
\DeclareMathOperator{\dd} {d}

\DeclareMathOperator{\rk} {rk}

\DeclareMathOperator{\ev}{ev}
\DeclareMathOperator{\lcm}{lcm}

\def\F{\mathbb{F}}
\def\Z{\mathbb{Z}}

\newcommand{\Fq}{\F_q}

\newcommand{\card}[1]{\left| #1 \right|}

\newcommand{\word}[1]{\ensuremath{\boldsymbol{#1}}}
\newcommand{\xv}{\word{x}}
\newcommand{\yv}{\word{y}}
\newcommand{\error}{\word{e}}
\newcommand{\av}{\word{a}}
\newcommand{\bv}{\word{b}}
\newcommand{\cv}{\word{c}}
\renewcommand{\error}{\word{e}}
\newcommand{\vv}{\word{v}}
\newcommand{\uv}{\word{u}}
\newcommand{\zv}{\word{z}}

\newcommand{\RS}[2]{\mathbf{RS}_{#1}[#2]}
\newcommand{\GRS}[2]{\mathbf{GRS}_{#1}[#2]}

\newcommand{\X}{\mathcal{X}}

\newcommand{\AGcode}[3]{\mathcal{C}_L (#1, #2, #3)}
\newcommand{\points}{\mathcal{P}}

\newcommand{\ip}[1]{\textcolor{red}{\bf [Isabella : #1 ]}}
\newcommand{\ac}[1]{\textcolor{blue}{\bf [Alain : #1 ]}}

\newcommand{\eqdef}{\stackrel{\textrm{def}}{=}}
\newcommand{\map}[4]{
    \left\{
    \begin{array}{ccc}
        #1 & \longrightarrow & #2 \\ #3 & \longmapsto & #4
    \end{array}
    \right.
}

\renewcommand{\leq}{\leqslant}
\renewcommand{\geq}{\geqslant}
\renewcommand{\le}{\leqslant}
\renewcommand{\ge}{\geqslant} 
\title{Power error locating pairs}
\author[1,2]{Alain Couvreur\thanks{\texttt{alain.couvreur@inria.fr}}}
\author[1,2]{Isabella Panaccione\thanks{\texttt{isabella.panaccione@inria.fr}}}
\affil[1]{Inria}
\affil[2]{
  LIX, CNRS UMR 7161\break
  École Polytechnique,\break
  91128 Palaiseau Cedex, France
}

\begin{document}

\maketitle

\begin{abstract}
  We present a new decoding algorithm based on error locating pairs
  and correcting an amount of errors exceeding half the minimum
  distance. When applied to Reed--Solomon or algebraic geometry codes,
  the algorithm is a reformulation of the so--called {\em power
    decoding} algorithm.  Asymptotically, it corrects errors up to
  Sudan's radius. In addition, this new framework applies to any code
  benefiting from an error locating pair. Similarly to Pellikaan's and
  K\"otter's approach for unique algebraic decoding, our algorithm
  provides a unified point of view for decoding codes with an
  algebraic structure beyond the half minimum distance. It permits to
  get an abstract description of decoding using only codes and linear
  algebra and without involving the arithmetic of
  polynomial and rational function algebras used for the definition of
  the codes themselves. Such algorithms can be valuable for instance
  for cryptanalysis to construct a decoding algorithm of a code
  without having access to the hidden algebraic structure of the code.
\end{abstract}

\medskip

\noindent {\bf Key words : } Error correcting codes; Reed--Solomon codes; algebraic geometry codes; decoding algorithms; power decoding; error correcting pairs;
cyclic codes.

\medskip

\noindent {\bf MSC : } 94B35, 94B27, 11T71,14G50.

\section*{Introduction}

Algebraic codes such as Reed--Solomon codes or algebraic geometry
codes are of central interest in coding theory because, compared to
random codes, these structured codes benefit from polynomial time
decoding algorithms that can correct a significant amount of errors.
The decoding of Reed--Solomon and algebraic geometry codes is a
fascinating topic at the intersection of algebra, algorithms, computer
algebra and complexity theory.

\subsection*{Decoding of Reed--Solomon codes}
Reed--Solomon codes benefit from an algebraic structure coming from
univariate polynomial algebras.  Thanks to this structure, one can
easily prove that they are maximum distance separable (MDS). In
addition, one can design an efficient unique decoding algorithm based
on the resolution of a so--called {\em key equation}
\cite{B68,B15,WB83,GS92}  and correcting up to half the minimum
distance.  This decoding algorithm is sometimes referred to as
Welch--Berlekamp algorithm in the literature.

In the late nineties, two successive breakthroughs due to Sudan
\cite{S97} and Guruswami and Sudan \cite{GS99} permitted to prove that
Reed--Solomon codes and algebraic geometry codes can be decoded in
polynomial time with an asymptotic radius reaching the so--called {\em
  Johnson bound} \cite{J62}. These algorithms have decoding radius
exceeding half the minimum distance at the price that they may return
a list of codewords instead of a single one. This drawback has
actually a very limited impact since in practice, the list size is
almost always less than or equal to $1$ (see \cite{M03} for further
details).  Note that decoding Reed--Solomon codes beyond the Johnson
bound remains a fully open problem: it is proved in \cite{GV05} that
the maximum likelihood decoding problem for Reed--Solomon codes is
NP--hard but the possible existence of a theoretical limit between the
Johnson bound and the covering radius under which decoding is possible
in polynomial time remains an open question with only partial answers
as in \cite{RW14}.

All the previously described decoders are {\em worst case}, i.e.
correct any corrupted codeword at distance less than or equal to some
fixed bound $t$. On the other side, some {\em probabilistic}
algorithms may correct more errors at the cost of some rare failures.
For instance, it is known for a long time that the classical
Welch--Berlekamp algorithm applied to interleaved Reed--Solomon is a
probabilistic decoder reaching the channel capacity \cite{SSB09} when
the number of interleaved codewords tends to infinity. Inspired by
this approach Bossert et. al. \cite{SSB10} proposed a probabilistic
decoding algorithm for decoding genuine Reed--Solomon codes by
interleaving the received word and some of its successive powers with
respect to the component wise product.  This algorithm has been called
{\em power decoding} in the sequel. A
striking feature of this power decoding is that it has the same
decoding radius as Sudan algorithm. Moreover an improvement of the
algorithm due to Rosenkilde \cite{N15} permits to reach Guruswami--Sudan radius, that is to say the Johnson bound.

However, compared to Sudan algorithm which is {\em worst case} and
returns always the full list of codewords at bounded distance from the
received word, the {\em power decoding} algorithm returns at most one
element and might fail. The full analysis of its failure probability
and the classification of failure cases is still an open
problem 
but practical experimentations give evidences that this failure
probability is very low.

\subsection*{Decoding of algebraic geometry codes}
All the previously described decoding algorithms for Reed--Solomon
codes have natural extensions to algebraic geometry codes at the cost
of a slight deterioration of the decoding radius which is proportional
to the curve's genus. The problem of decoding algebraic geometry codes
motivated hundreds of articles in the last three decades. The story
starts in the late 80's with an article of Justesen, Larsen, Jensen,
Havemose and H{\o}holdt \cite{JLJHH89} proposing a syndrome based
decoding algorithm for codes from plane curves. The algorithm has then
been extended to arbitrary curves by Skorobogatov and Vl\u{a}du\c{t}
\cite{SV90}.  The original description of these algorithms was
strongly based on algebraic geometry. However, subsequently, Pellikaan
\cite{P88, P92} and independently K\"otter \cite{K92} proposed an
abstract description of these algorithms expurgated from the formalism
of algebraic geometry. This description was based on an object called
{\em error correcting pair} . An error correcting pair for a code $C$
is a pair of codes $A, B$ satisfying some dimension and minimum
distance constraints and such that the space $A * B$ spanned by the
component wise products of words of $A$ and $B$ is contained in
$C^\perp$. The existence of such a pair 
provides an efficient decoding algorithm essentially based on linear algebra.
This approach provides a unified framework to describe algebraic
decoding for Reed--Solomon codes, algebraic geometry codes and some
cyclic codes \cite{D93, DK94}.

Concerning list decoding, Sudan and Guruswami--Sudan algorithms extend
naturally to algebraic geometry codes.  Sudan algorithm has been
extended to algebraic geometry codes by Shokrollahi and Wasserman
\cite{SW99} and Guruswami--Sudan original article \cite{GS99} treated
the list decoding problem for both Reed--Solomon and algebraic
geometry codes.  Similarly, the power decoding algorithm generalises
to such codes. As far as we know, no reference provides such a
generalisation in full generality, however, such an algorithm is
presented for the case of Hermitian curves in \cite{NB15} and the
generalisation to arbitrary curves is rather elementary and
sketched in Appendix~\ref{sec:appendix}.

For further details on unique and list decoding algorithms
for algebraic geometry codes, we refer the reader to the excellent
survey papers \cite{HP95,BH08}.

\subsection*{Our contribution}
In summary, on the one hand, Reed--Solomon and algebraic geometry
codes benefit from a Welch--Berlekamp like unique decoding algorithm,
a Sudan--like list decoding algorithm and a power decoding. On the
other hand, an abstract unified point of view for the unique decoding
of these codes is given by error correcting pairs.  In the present
article we partially fill the gap by proposing a new algorithm in the
spirit of error correcting pairs algorithms but permitting to correct
errors beyond half the minimum distance. For Reed--Solomon and
algebraic geometry codes, our algorithm is nothing but an abstraction
of the power decoding algorithm, exactly as Pellikaan's error
correcting pairs algorithm is an abstraction of
Welch--Berlekamp. However, this abstract version applies to any code
equipped with a power error locating pair such as many cyclic codes
\cite{D93, DK94}, for which our algorithm permits to correct errors
beyond the Roos bound (\cite{R83}). 

Interestingly, for algebraic geometry codes, the analysis of the
number of errors our algorithm can correct turns out to provide a
slightly better decoding radius compared to the analysis of the power
decoding. According to experimental observations, the radius we obtain
seems out to be optimal for both algorithms.

In addition to the intrinsic interest of having a unified point of
view for decoding beyond half the minimum distance of codes equipped
with a power error locating pair, a possible application of such an
approach lies in cryptography and in particular in cryptanalysis.
Indeed, it has been proved in \cite{CMP17}, that an 
error correcting
pair can be recovered from the single knowledge of a generator matrix
of an algebraic geometry code. Using the present article's results,
this approach can be extended in order to build a Power Error Locating Pair
from the very knowledge of a generator matrix of the code.
This shows that any
McEliece like scheme based on algebraic geometry codes and using a
Sudan--like decoder for the decryption step cannot be secure since
it is possible to recover in polynomial time all the necessary data to
decode efficiently from the single knowledge of the public key.

\begin{figure}[h]
  \centering
  \begin{tikzpicture}
    \node (A){$t\le\Bigl\lfloor\frac{d-1}{2}\Bigr\rfloor$};
    \node (B)[node distance=3.5cm, right of=A]{Welch-Berlekamp};
    \node (C)[node distance=4.5cm, right of=B]{Error correcting pairs};
    \node (F)[node distance=2cm, below of=A]{$t> \Bigl\lfloor\frac{d-1}{2}\Bigr\rfloor$};
    \node (D)[node distance=2cm, below of=B]{Sudan};
    \draw[->,font=\scriptsize] (B) to node [right]{} (D);
    \node(G)[node distance=0.5cm, below of=D]{Power decoding};
    \node (E)[node distance=2cm, below of=C]{?};
    \draw[->,font=\scriptsize] (C) to node [right]{} (E);
  \end{tikzpicture}
  \caption{Existing decoding algorithms for Reed--Solomon and algebraic
    geometry codes}
  \label{fig:existing}
\end{figure}
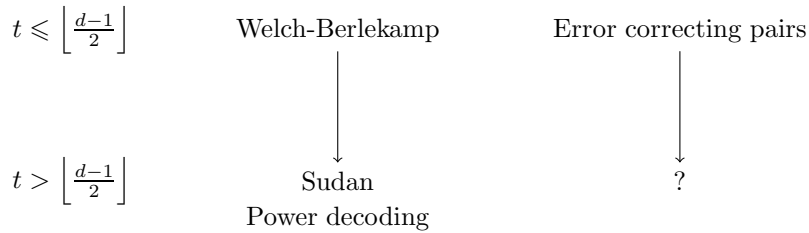

\subsection*{Outline of the article}
Notation and prerequisites are recalled in Section~\ref{sec:nota}.  In
Section~\ref{section1}, we present some well--known decoding
algorithms for algebraic codes with a particular focus on
Reed--Solomon codes. Namely, a presentation of Welch--Berlekamp
algorithm, power decoding and error correcting pairs algorithm is
given. Our contribution, the {\em power error locating pairs}
algorithm is presented in Section~\ref{secPECPRS}. Next, we explain
how this abstract algorithm can be applied to Reed--Solomon codes in
Section~\ref{sec:PELP_RS} and to algebraic geometry codes in
Section~\ref{section:AGcodes}. The latter section is concluded with a
discussion on cryptanalytic applications of power error locating
pairs.  Finally, our algorithm is applied to cyclic codes in
Section~\ref{sec:cyclic}.

 \section{Notation and prerequisites}\label{sec:nota}
\subsection{Basic notions of coding theory}
In this article, we work over a finite field $\F_q$. Given a positive
integer $n$, the {\em Hamming weight} on $\F_q^n$ is denoted as
\[
  \forall \xv = (x_1, \ldots, x_n) \in \F_q^n,\quad
  \w(\xv) \eqdef \card{\{i \in \{1, \ldots, n\} ~|~ x_i \neq 0\}}
\]
and the {\em Hamming distance} $\dd(\ \! \cdot\  ,\ \! \cdot\ \!)$ as
\[
  \forall \xv, \yv \in \Fq^n,\quad \dd(\xv, \yv) \eqdef \w (\xv - \yv).
\]
The {\em support} of a vector $\av \in \F_q^n$ is the subset
$\supp (\av) \subseteq \{1, \ldots, n\}$ of indexes $i$ satisfying
$a_i \neq 0$.

All along this article, term {\em code} denotes a {\em linear code},
i.e. a linear subspace of $\Fq^n$.  Given a code $C \subseteq \Fq^n$,
the minimum distance of $C$ is denoted by $\dd (C)$ or $\dd$ when
there is no ambiguity on the code.  In addition, we recall some
classical notions of coding theory and introduce notation for them.

\begin{defn}\label{def:short_punct}
  Given $J=\{j_1,\dots, j_s\}\subset \{1,\dots, n\}$ and
  $\xv=(x_1,\dots,x_n)\in~\mathbb{F}_q^n$, we denote by $\xv_J$ the
  projection of $\xv$ on the coordinates in $J$ and by $Z(\xv)$
  the complement of the support of $\xv$ in $\{1, \dots, n\}$.
  Namely,
  \[
    \xv_J\eqdef(x_{j_1},\dots, x_{j_s}) \qquad {\rm and} \qquad
    Z(\xv)\eqdef \{i\in\{1,\dots, n\}\mid x_i=0\}
  \]
Moreover, for $A\subseteq\mathbb{F}_q^n$, we define
\begin{enumerate}[$(i)$] 
\item $A_J\eqdef\{\av_J\mid \av\in A\}\subseteq\mathbb{F}_q^{|J|}$
  (\textit{puncturing});
\item
  $A(J)\eqdef\{\av\in A\mid \av_J=\textbf{0}\}\subseteq\mathbb{F}_q^n$
  (\textit{shortening});
\item $Z(A)\eqdef\{i\in\{1,\dots, n\}\mid a_i=0\ \ \forall \av\in A\}$
  ({\it complementary of the support}).
\end{enumerate}
\end{defn}

\begin{rem}
  In (\textit{ii}) we made an abuse of language by using the term
  \textit{shortening} which classically refers to the operation
  \[
    \{\av_{\{1,\dots, n\}\setminus J}\mid \av_J=
    \textbf{0}\}.
  \]
In comparison to the usual definition we do not remove the prescribed
zero positions. We do so because we need to compute some star products
(see \S\ref{subsec:star_prod}) and for this sake involved vectors
should have the same length.
\end{rem}

\begin{rem}
  Classically in the literature, {\em puncturing} a code at a subset
  $J$ means deleting the entries whose index is in $J$. Here our notation
  does the opposite operation by keeping only the positions whose indexes
  are in $J$ and removing all the other ones.
\end{rem}

\subsection{Decoding problems}
The purpose of the present article is to introduce a new algorithm
extending the {\em error correcting pairs algorithm} (see
\S~\ref{sec:ECP}) and permitting to correct errors beyond half the
minimum distance.  This is possible at one cost: the algorithm might
fail sometimes.  First, let us state the main decoding problem one
wants to solve.

\begin{prob}[Bounded decoding]\label{pb:bd_dec}
  Let $C \subseteq \Fq^n$, $\yv \in \Fq^n$ and $t \in \{0, \ldots, n\}$.
  Return (if exists) a word $\cv \in C$ such that $\dd(\cv, \yv) \leq t$.  
\end{prob}

\begin{defn}
For an algorithm solving Problem~\ref{pb:bd_dec}, the largest possible
$t$ such that the algorithm succeeds is referred to as the {\em decoding
radius} of the algorithm.  
\end{defn}

When $t \leq \left\lfloor {(\dd(C)-1)}/{2} \right\rfloor$, then the
solution, if exists, is unique and the corresponding problem is
sometimes referred to as the {\em unambiguous decoding problem}.  For
larger values of $t$, the set of codewords at distance less than or
equal to $t$ from $\yv$ might have more than one element.  To solve
the bounded decoding problem in such a situation, various decoders
exist. Some of them return the closest codeword (if unique), while other ones
return the whole list of codewords at distance less than or equal to
$t$. The algorithm of the present article, either returns a unique
solution or fails.

To conclude this section, let us state an assumption which we suppose
to be satisfied for any decoding problems considered in the sequel.

\begin{assumption}\label{as:a_solution_exists}
  In the following, given a code $C$ and a positive integer $t$, when
  considering Problem~\ref{pb:bd_dec}, we always suppose
  that the {\em received vector} $\yv \in \Fq^n$ is of the form
  $ \yv = \cv + \error $ where $\cv \in C$ and $\w (\error) = t$.
  Equivalently, we always suppose that the bounded decoding problem
  has at least one solution.
\end{assumption}

\subsection{Reed--Solomon codes}\label{subsec:RS}
The space of polynomials with coefficients in $\Fq$ and degree less
than $k$ is denoted by $\F_q[X]_{<k}$.  Given an integer $n\ge k$ and
a vector $\xv \in \Fq^n$ whose entries are distinct, the {\em
  Reed--Solomon code} of length $n$ and dimension $k$ is the image of
the space $\F_q[X]_{<k}$ by the map
\begin{equation}\label{eq:ev_map}
  \ev_{\xv} : \map{\Fq [X]}{\Fq^n}{f}{(f(x_1), \ldots, f(x_n)).}
\end{equation}
This code is denoted by $\RS{q}{\xv, k}$ or $\RS{q}{k}$ when
there is no ambiguity on the vector $\xv$.
That is:
\[
  \RS{q}{k} \eqdef \left\{ (f(x_1), \ldots, f(x_n)) ~|~
  f \in \Fq[X]_{<k}\right\} = \ev_{\xv} \left(\F_q[X]_{<k} \right).
\]
One can actually consider a larger class of codes called
\textit{generalised Reed--Solomon codes} and defined as:
\[\GRS{q}{\xv, \yv, k}\eqdef \left\{ (y_1f(x_1), \ldots, y_nf(x_n)) ~|~
    f \in \Fq[X]_{<k}\right\},\] where
$\yv\in{(\mathbb{F}_q^\times)}^n$. Such a code has length $n$,
dimension $k$ and minimum distance $\dd = n - k + 1$. In this article,
for the sake of simplicity, we focus on the case of Reed--Solomon codes
($\yv=(1,\ldots, 1)$) with $n = q$, i.e. the so--called {\em
  full--support} Reed--Solomon codes. This context is much more
comfortable for duality since we can assert that
\[
  \RS{q}{k}^\perp = \RS{q}{n-k}.
\]
In the general case, the above statement remains true by replacing
Reed--Solomon codes by generalised Reed--Solomon codes with a specific
choice of $\yv$. Indeed, it holds
$\GRS{q}{\xv, \yv, k}^{\perp}=\GRS{q}{\xv, \yv', n-k}$, where
\[\yv'=-\frac{1}{y_i\prod_{j=1\atop j\ne i }^n(x_j-x_i)}\cdot\]
See for instance \cite[Problem~5.7]{R06}.
Actually, the results
of the present article extend straightforwardly to generalised Reed--Solomon
codes at the cost of slightly more technical proofs.

\subsection{Algebraic geometry codes}
In what follows, by {\em curve} we always mean a smooth projective
geometrically connected curve defined over $\Fq$.
Given such a curve $\X$, a divisor $G$ on $\X$
and a sequence $\points = (P_1, \dots, P_n)$ of rational points of
$\X$ avoiding the support of $G$, one can define the code
\[
  \AGcode{\X}{\points}{G} \eqdef \left\{
    (f(P_1), \dots, f(P_n)) ~|~ f\in L(G)
  \right\},
\]
where $L(G)$ denotes the Riemann--Roch space associated to $G$.  When
$\deg G < n$, such a code has dimension $k \geq \deg G +1 - g$ where
$g$ denotes the genus of $\X$ and minimum distance $\dd > n - \deg G$.
We refer the reader to \cite{TVN07,S09} for further details
on algebraic geometry, function fields and algebraic
geometry codes.

\subsection{Star product of words and codes}\label{subsec:star_prod}
The space $\Fq^n$ is a product of $n$ fields and hence has a natural
structure of ring. We denote by $*$ the component wise product
of vectors
\[(a_1, \dots, a_n) * (b_1, \dots, b_n) \eqdef (a_1 b_1, \dots , a_n
  b_n).
\]
Given a vector $\av \in \F_q^n$, the $i$--th power $\av^i$ of $\av$ is
defined as $\av^i \eqdef (a_1^i, \dots, a_n^i)$.

\begin{rem}\label{propstarprod}
  This product should not be confused with the canonical inner product
  in $\mathbb{F}_q^n$ 
  defined by $\langle \av, \bv \rangle = \sum_{i=1}^n a_i b_i$.
  Note that these two operations are related by the following
  adjunction property
  \begin{equation}\label{eq:adjunction}
    \langle \av\ast \bv,\cv \rangle=\langle \av, \bv\ast\cv\rangle.
  \end{equation}
In particular
$\av\ast \bv\in \cv^{\perp}\iff \av\ast \cv\in \bv^{\perp}$.
\end{rem}

\begin{defn}
  Given two codes $A, B \subseteq \Fq^n$, the star product $A * B$ is
  the code spanned by all the products $\av * \bv$ for $\av \in A$ and
  $\bv \in B$. If $A = B$, the product is denoted by $A^2$.
  Inductively, one defines $A^i \eqdef A * A^{i-1}$ for
  all $i\geq 2$.
\end{defn}

\subsubsection{Star products of algebraic codes}
 Algebraic codes satisfy peculiar properties with respect
to the star product.

\begin{prop}[Star product of Reed--Solomon codes]\label{prop:star_prod_GRS}
  Let $\xv \in \Fq^n$ be a vector with distinct entries and $k, k'$
  be positive integers such that $k+k'-1\le n$. Then,
  \[
    \RS{q}{\xv, k} * \RS{q}{\xv, k'} = \RS{q}{\xv, k+k'-1}.
  \]
\end{prop}

\begin{rem}\label{rem:star_product_for_RS}
  Actually Proposition~\ref{prop:star_prod_GRS} holds true even if 
$k+k'-1> n$ but in this situation the right--hand side becomes $\F_q^n$.
\end{rem}

\begin{prop}[Star product of AG codes]\label{prop:star_prod_AGcodes}
  Let $\X$ be a curve of genus $g$, $\points = (P_1, \dots, P_n)$
  be a sequence of rational points of $\X$ and $G, G'$ be two divisors of
  $\X$ such that $\deg G \geq 2g$ and $\deg G' \geq 2g+1$. Then,
  \[
    \AGcode{\X}{\points}{G} * \AGcode{\X}{\points}{G'}
    = \AGcode{\X}{\points}{G+G'}.
  \]
\end{prop}

\begin{proof}
  This is a consequence of \cite[Theorem 6]{M70}.
  For instance, see \cite[Corollary 9]{CMP17}.
\end{proof}

\subsubsection{A Kneser--like theorem}\label{sss:Kneser}
We conclude this section with a result which will be useful in the sequel
and can be regarded as a star product counterpart of the famous Kneser Theorem
in additive combinatorics (see \cite[Theorem~5.5]{TV06}).
We first have to introduce a notion.

\begin{defn}
  Let $H \subseteq \Fq^n$ be a code, the {\em stabiliser} of $H$ is defined
  as
  \[
    \textrm{Stab}(H) \eqdef \{\xv \in \Fq^n ~|~ \xv \ast H \subseteq H\}.
  \]
\end{defn}

\begin{thm}\label{thm:kneser}
  Let $A, B \subseteq \Fq^n$ be two codes. Then,
  \[
    \dim(A \ast B) \geq \dim A + \dim B - \dim \textrm{Stab}(A \ast B).
  \]
\end{thm}

\begin{proof}
  See \cite[Theorem 18]{MZ15} or \cite[Theorem 4.1]{BL17}.
\end{proof}

For any code $C$, the stabiliser of $C$ has dimension at least $1$
since it contains the span of the vector $(1, \ldots, 1)$. On the
other hand, it has been proved in \cite[Theorem 1.2]{KS80} that a code
$C$ has a stabiliser of dimension $>1$ if and only if it is {\em
  degenerated}, i.e. if and only if either it is a direct sum of
subcodes with disjoint supports or any generator matrix of $C$
has a zero column.
This leads to the following analog of Cauchy--Davenport
Theorem (\cite[Theorem~5.4]{TV06}).

\begin{cor}\label{cor:CauchyDavenport}
  Let $A, B \subseteq \Fq^n$ be two codes such that $A \ast B$ is non
  degenerated, then
  \[
    \dim A \ast B \geq \dim A + \dim B - 1.
  \]
\end{cor}

 \section{Former decoding algorithms for Reed--Solomon and
  algebraic geometry codes}\label{section1}

It is known that several different decoding algorithms have been
designed for Reed--Solomon and algebraic geometry codes. In particular,
depending on the algorithm, we are able to solve either
Problem~\ref{pb:bd_dec} up to half the minimum distance, or
to the Johnson bound.

In this section, we recall several decoding algorithms for
Reed--Solomon codes. For all of them, a natural extension to algebraic
geometry codes is known.
Recall that, whenever we
discuss a decoding problem we suppose
Assumption~\ref{as:a_solution_exists} to be satisfied, i.e.  we
suppose that the decoding problem has at least one solution.  Hence,
we can write
\begin{equation}\label{forme_yv}
\yv=\cv+\error,
\end{equation}
for some $\cv\in C$ and $\error\in\mathbb{F}_q^n$ with
$\w(\error)=t$. Note that since $C=\RS{q}{k}$, the codeword $\cv$ can
be written as the evaluation of a polynomial $f(x)$ with
$\deg(f)<k$. The vector $\error$ is called the \textit{error
  vector}. Moreover, we define
\[I_{\error}\eqdef\supp(\error)= \{i\in\{1, \dots, n\}\mid e_i\ne 0\}.\]
Hence, we have $t=\w(\error)=|I_{\error}|$.

\subsection{Welch--Berlekamp algorithm}
Welch--Berlekamp algorithm boils down to a linear system based on $n$
key equations. In this case the decoding radius is given by a
sufficient condition, i.e. if there is a solution, the algorithm will
find it.
\begin{defn}
Given $\yv, \error$ and $I_{\error}$ as above we define 
\begin{itemize}
\item the \textit{error locator polynomial} as
  $\Lambda(X)\eqdef\prod_{i\in I_{\error}} (X-x_i)$;
\item $N(X)\eqdef \Lambda(X)f(X)$.
\end{itemize}
\end{defn}

\medskip\noindent Hence, for any $i\in\{1,\dots, n\}$, the polynomials
$\Lambda$ and $N$ verify
\begin{equation}\label{keyeqBW}
\Lambda(x_i)y_i=N(x_i).
\end{equation}
The aim of the algorithm is then to solve the following
\begin{kprob}\label{probberl}
  Find a pair of polynomials $(\lambda, \nu)$ such that
  $\deg(\lambda)\le t$, $\deg(\nu)\le~t+k-1$ and
\begin{equation}\label{eq:SBW}\tag{$S_{\textrm{WB}}$}
  \forall i\in \{1,\dots n\},\quad \lambda(x_i)y_i=\nu(x_i).
\end{equation}
\end{kprob}

\begin{rem}
Actually, the degrees of $\Lambda$ and $N$ are related by 
\[ \deg(N)\le \deg(\Lambda)+k-1.\] With this constraint the problem
can be solved using Berlekamp--Massey algorithm.  Though in this paper
we chose to consider the simplified constraints of Key Problem
\ref{probberl} in order to have linear constrains, which makes the
analysis of the decoding radius easier. By the following lemma it will
be clear that making this choice has no consequence on the decoding radius
of Welch--Berlekamp algorithm.
\end{rem} 
\medskip  The system \eqref{eq:SBW} is linear and has $n$ equations
in $2t+k+1$ unknowns. We know that the pair $(\Lambda, \Lambda f)$
is in its solutions space. The following result proves that, for
certain values of $t$, actually it is not necessary to find exactly
that solution to solve the decoding problem.

\begin{lem}\label{lemBW}
  Let $t\le \frac{\dd-1}{2}$. If $(\lambda, \nu)$ is a nonzero solution
  of \eqref{eq:SBW}, then $\lambda \neq 0$ and $f=\frac{\nu}{\lambda}$.
\end{lem}

\medskip\noindent For the proof we refer for instance to \cite[Theorem
4.2.2]{JH04}. We can finally write the algorithm (see Algorithm
1). Its correctness is entailed by Lemma~\ref{lemBW} whenever
$t\le \frac{\dd-1}{2}$, that is, the decoding radius of Welch--Berlekamp
algorithm is $t=\lfloor\frac{\dd-1}{2}\rfloor\cdot$
\begin{algorithm}\label{AlgoBW}
\caption{Welch-Berlekamp Algorithm}
\textbf{Inputs:} $\yv\in \mathbb{F}_q^n$ as in (\ref{forme_yv}),
$k\le n, t=\w(\error)\le\frac{\dd-1}{2}$.\\
\textbf{Output:} $f\in\mathbb{F}_q[X]_{<k}$ such that
$\dd(\ev_{\xv}(f),\yv)=t$.
\begin{algorithmic}[1]
  \State $(\lambda, \nu)\gets$ arbitrary nonzero element in the solutions
  space of \eqref{eq:SBW}
  \State \Return $f=\frac{\nu}{\lambda}$
\end{algorithmic}
\end{algorithm}
\noindent

\subsection{Power decoding algorithm}

Introduced by Sidorenko, Schmidt and Bossert, \cite{SSB09}, {\em power
  decoding} is inspired from a decoding algorithm for interleaved
Reed--Solomon codes. It consists in considering several ``powers''
(with respect to the star product) of the received vector $\yv$ in
order to have more relations to work on. Given the vector
$\yv=\cv+\error$ we want to correct, we consider the $i$--th powers
$\yv^i$ of $\yv$ for $i=1,\dots, \ell$ (see \S~\ref{subsec:star_prod}
for the definition of $\yv^i$). In this section, we only present the
case $\ell=2$ for simplicity.  We have
\begin{align}
\yv=\ & \cv+\error \label{eq:c+e}\\
\yv^2=\ & \cv^2+2\cv\ast \error+ \error^2\label{eq:c2+e2}.
\end{align}
We then rename $\error$ by $\error^{(1)}$ and
$2\cv\ast \error+ \error^2$ by $\error^{(2)}$ and get
\begin{align}
  \yv &= \cv+\error^{(1)}\label{eq:c+eprime}\\
  \yv^2 &= \cv^2+\error^{(2)}\label{eq:c2+e2prime}.
\end{align}
One can see $\yv^2$ as a perturbation of a word
$\cv^2\in C^2= \RS{q}{2k-1}$ by the error vector $\error^{(2)}$. Hence
$\yv^2$ is an instance of another decoding problem.
In addition, we have the following elementary result which is the key
of power decoding.

\begin{prop}\label{remsupp}
It holds $\supp(\error^{(2)})\subseteq \supp(\error^{(1)})$. 
\end{prop}

\noindent
It asserts that on $\yv$ and $\yv^2$, the errors are localised at the
same positions. More precisely, error positions on $\yv^2$ are error
positions on $\yv$.  Hence, we are in the error model of interleaved
codes: equations (\ref{eq:c+eprime}) and (\ref{eq:c2+e2prime}) can be
regarded as a decoding problem for the interleaving of two codewords
with errors at most at $t$ positions. Therefore, errors can be decoded
simultaneously using the same error locator polynomial.  To do that,
we consider the error locating polynomial
$\Lambda(X)=\prod_{i\in I_{\error}}(X-x_i)$ as for Welch--Berlekamp
algorithm and the polynomials
\[
  N_1\eqdef \Lambda f,\quad N_2\eqdef\Lambda f^2.
\]
Thanks to Proposition~\ref{remsupp}, it is possible to write the key
equations
\begin{equation}\label{Keyeqpow}
  \left\{
    \begin{aligned}
      \Lambda(x_i)y_i= N_1(x_i)\quad &\forall i\in \{1,\dots, n\}\\
      \Lambda(x_i)y_i^2= N_2(x_i)\quad &\forall i\in \{1,\dots, n\}
    \end{aligned}
  \right.
\end{equation}
Consequently, the Power decoding algorithm consists in solving the
following problem.
\begin{kprob}\label{probpow}
  Given $\yv\in\mathbb{F}_q^n$ and $t\in\mathbb{N}$, find
  $(\lambda, \nu_1, \nu_2)$ which fulfill
\begin{equation}\label{eq:Spo2}\tag{$S_{\textrm{Po}}$}
  \left\{
    \begin{aligned}
      \lambda(x_i)y_i=\nu_1(x_i),\quad &\forall i\in \{1,\dots, n\}\\
      \lambda(x_i)y_i^2= \nu_2(x_i),\quad &\forall i\in \{1,\dots, n\}.
    \end{aligned}
  \right.
\end{equation}
with $\deg(\lambda)\le t$ and $\deg(\nu_j)\le t+j(k-1)$ for
$j\in \{1, 2\}$.
\end{kprob}
 
\begin{rem}
  Key Problem \ref{probpow} is slightly different from the problem
  faced in the original paper describing Power Decoding
  (\cite{SSB10}). We used a \textit{key equation} formulation of the
  problem instead of the \textit{syndrome} one. The two formulations
  are equivalent if the right bounds on polynomials' degrees are
  taken (see \cite{R15}). In particular, one should look for
  $(\lambda, \nu_1, \nu_2)$ such that
  $\deg(\nu_j)\le\deg(\lambda)+j(k-1)$ for all
  $j\in \{1, \dots, \ell\}$.  However, similarly as for
  Welch--Berlekmap algorithm, we consider two weaker constraints which
  allow to reduce the problem to a linear system to solve. The price
  is that our Key Problem could get more failure cases than problem in
  \cite{SSB10}. However, we observed experimentally that these cases
  are really rare.
\end{rem}

\medskip\noindent The vector $(\Lambda, \Lambda f, \Lambda f^2)$ is a
solution of the linear system \eqref{eq:Spo2}. Though, at the moment
there is no guaranteed method to recover it among all the 
solutions. We only know that, if $g$ is a polynomial such that
$\deg(g)<k$ and $\dd(\ev_{\xv}(g), \yv)\le t$, then there exists an
error locator polynomial $\Gamma$ of the error with respect to $g$,
such that the vector
\begin{equation}\nonumber
(\Gamma, \Gamma g, \Gamma g^2)
\end{equation}
is solution of \eqref{eq:Spo2}. Among all solutions like that, we want
to pick the one that gives the closest codeword, that is the one such
that $\deg(\Gamma)$ is minimal (see pt.1 in Algorithm 2).
\begin{algorithm}\label{AlgoPD}
\caption{Power decoding algorithm with $\ell=2$}
\textbf{Inputs:} $\yv\in\mathbb{F}_q^n$ as in (\ref{forme_yv}),
$t=\w(\error), k\le n$\\
\textbf{Output:} some $g\in\mathbb{F}_q[X]_{<k}$ such that
$\dd(\ev_{\xv}(g), \yv)\le t$ or failure
\begin{algorithmic}[1]
  \State{$(\lambda, \nu_1, \nu_2)\gets$ a nonzero solution of
    \eqref{eq:Spo2} with $\lambda$ of smallest possible degree.}
  \If{$\left(\lambda|\nu_1 \ \ \mathbf{and}\ \ \lambda| \nu_2 \ \
      \mathbf{and} \ \
      {\left(\frac{\nu_1}{\lambda}\right)}^2=\frac{\nu_2}{\lambda}\right)$}
    $g\eqdef\frac{\lambda}{\nu_1}$ \If{$\dd(\ev_{\xv}(g), \yv)\le
      t$ \textbf{and} $\deg(g)<k$,}
  \State{\Return $g$.}
  \EndIf
  \EndIf \State \Return failure
\end{algorithmic}
\end{algorithm}

\begin{rem}
  The $2n$ equations we obtain in Key Problem \ref{probpow}, consist
  in the key equations for $\yv$ and the key equations for $\yv^2$,
  that is two simultaneous decoding problems. Indeed, the important
  aspect is that these two decoding problems share the error locator
  polynomial $\Lambda$. Hence, by adding $n$ relations, we only add
  $\deg (N_2)+1$ unknowns instead of $\deg(N_2)+t+2$.
\end{rem}

\begin{rem} To compute the decoding radius of the Power
  Decoding algorithm we look for a condition on the size of the system
  \eqref{eq:Spo2}. Note that the algorithm gives one solution or none,
  hence there cannot be a sufficient condition for the correctness
  of the algorithm as soon as $t>\frac{\dd-1}{2}$.
  For this reason, we look for a {\bf necessary condition} for the system
  \eqref{eq:Spo2} to have a solution space of dimension $1$.
\end{rem}

\paragraph{The general case with an arbitrary $\ell$.}
For an arbitrary $\ell$, Key Problem~\ref{probpow} is replaced the following
system:
\begin{equation}\label{eq:Spo}\tag{$S_{\textrm{Po}}$}
  \left\{
    \begin{aligned}
      \lambda(x_i)y_i=\nu_1(x_i),\quad &\forall i\in \{1,\dots, n\}\\
      \vdots \qquad & \\
      \lambda(x_i)y_i^\ell= \nu_\ell(x_i),\quad &\forall i\in \{1,\dots, n\}.
    \end{aligned}
  \right.
\end{equation}
with $\deg \lambda \leq t$ and $\deg \nu_j \leq t + j(k-1)$ for all $j \in
\{1, \ldots, \ell\}$.

\subsubsection{Decoding radius}
Under Assumption~\ref{as:a_solution_exists}, we know that
$(\Lambda, \Lambda f, \Lambda f^2)$ is a solution for \eqref{eq:Spo2}
and hence, the algorithm returns $\cv$ if and only if the solution space of
\eqref{eq:Spo2} has dimension one. 

Let us define the polynomial
\begin{equation}\nonumber
\pi(X)\eqdef\prod_{i=1}^n (X-x_i)
\end{equation}
and consider the bounds in Key Problem~\ref{probpow} on the degrees of
$\nu_1$ and $\nu_2$. If $t+2(k-1)\ge n$, then we would have
$
(0, 0, \pi)\in Sol
$. 
For a larger $\ell$ this condition becomes
\begin{equation}\label{nec_cond: Pow_0}
t<n-\ell(k-1).
\end{equation}
However, bound (\ref{nec_cond: Pow_0}) is actually much larger than
the decoding radius. We look then for a stricter bound on $t$. Another
necessary condition to have a solution space of dimension one
for~(\ref{eq:Spo}), is:
\[\#unknowns\le \#equations +1,\]
which gives $t\le \frac{2n-3k+1}{3}\cdot$

\begin{rem}
  Actually, starting from $\yv = \cv + \error$, Algorithm 2
  could return another word $\cv' \neq\cv$ which would be closer to
  $\yv$. In such a situation, the solution space of (\ref{eq:Spo2})
  will not have dimension $1$ since it will contain a triple
  $(\Lambda, \Lambda f, \Lambda f^2)$ associated to $\cv$
  and a vector $(\Lambda', \Lambda' f', \Lambda' f'^2)$ associated to $\cv'$
  which can be proved to be linearly independent.
  Therefore, the algorithm may return the closest codeword even in situations
  when the dimension of the solutions space of~(\ref{eq:Spo2}) has dimension
  larger than $1$.
  The analysis consists in giving a necessary condition for the algorithm to
  return $\cv$.
\end{rem}

Finally, the same process can be
used for a general $\ell$ and we obtain the following decoding radius
\begin{equation}\label{decradPow}
t\le \frac{2n\ell-k\ell(\ell+1)+\ell(\ell-1)}{2(\ell+1)}\cdot
\end{equation}

\subsection{The error correcting pairs algorithm}\label{sec:ECP}
The Error Correcting Pairs (ECP) algorithm has been designed by
Pellikaan \cite{P92} and independently by K\"otter \cite{K92}. Its
formalism gives an abstract description of a decoding algorithm
originally arranged for algebraic geometry codes \cite{SV90} and whose
description required notions of algebraic geometry. In their works,
Pellikaan and K\"otter, simplified the instruments needed in the
original decoding algorithm and made the algorithm applicable to any
linear code benefiting from a certain elementary structure called
\textit{error correcting pair} and defined in Definition~\ref{ecp}
below.  Given a code $C$ and a received vector $\yv = \cv + \error$
where $\cv \in C$ and $\w (\error) \leq t$ for some positive integer
$t$, the ECP algorithm consists in two steps:
\begin{enumerate}[(1)]
\item\label{item:stepECP_1} find $J\subseteq\{1, \dots, n\}$ such that
  $J\supseteq I_{\error}$, where $I_{\error}$ denotes the support of $\error$;
\item\label{item:stepECP_2} recover the nonzero entries of $\error$.
\end{enumerate}
As said before, these steps can be solved if the code has a $t$-error
correcting pair where $t=\w(\error)$ is small enough.
\begin{defn}\label{ecp}
  Given a linear code $C\subseteq \mathbb{F}_q^n$, a pair $(A, B)$ of
  linear codes, with $A, B\subseteq \mathbb{F}_q^n$ is called
  $t$-\textit{error correcting pair} for $C$ if
\begin{enumerate}[(ECP1)]
\item\label{item:ECP1} $A\ast B\subseteq C^{\perp}$; 
\item\label{item:ECP2} $\dim(A)>t$;
\item\label{item:ECP3} $\dd(B^{\perp})>t$;
\item\label{item:ECP4} $\dd(A)+\dd(C)>n$.
\end{enumerate}
\end{defn}
\begin{rem}\label{rem: swich}
One can observe that, thanks to Remark \ref{propstarprod},
\begin{equation}\nonumber
A\ast B\subseteq C^{\perp}\iff A\ast C\subseteq B^{\perp}.
\end{equation}
\end{rem}
\medskip\noindent Since this notion does not look very intuitive, an
example of error correcting pair for Reed--Solomon codes is given
further in \S~\ref{secECPRS} and an interpretation of the various
hypotheses above is given in light of this example in \S~\ref{subsec:ECPRS}.
For now, we want to explain more precisely how the ECP
algorithm works.

\subsubsection{First step of the error correcting pair algorithm}
In Step (\ref{item:stepECP_1}) of the ECP algorithm, we wish to find a
set which contains
$I_{\error}$. 
A good candidate for $J$ could then be $Z(A({I_{\error}}))$, indeed
the following result (see \cite{P92}), entails that
${I_{\error}}\subseteq Z(A({I_{\error}}))\varsubsetneq \{1,\dots, n\}$
(see Definition~\ref{def:short_punct}).
\begin{prop}\label{dim(A)}
If $\dim(A)> t$, then $A({I_{\error}})\ne 0$.
\end{prop}

\medskip\noindent Though, since we do not know ${I_{\error}}$, we do
not have any information about $A({I_{\error}})$. That is why a new
vector space is introduced:
\begin{equation}\label{eq:M}
  M_1\eqdef\{\av\in A\mid \langle \av\ast \yv,\bv\rangle=0 \ \ \forall
  \bv\in B\}.
\end{equation}
The key of the algorithm is in the following result.

\begin{thm}\label{thmecp}
  Let $\yv = \cv + \error$, ${I_{\error}}=\supp(\error)$ and $M_1$ as
  above. If $A\ast B\subseteq C^{\perp}$, then
\begin{enumerate}[(1)]
\item $A({I_{\error}})\subseteq M_1\subseteq A$;
\item if $\dd(B^{\perp})>t$, then $A({I_{\error}})=M_1$;
\end{enumerate} 
\end{thm}

\begin{proof} See \cite{P92}.
\end{proof}

\medskip\noindent Therefore, if the pair $(A, B)$ fulfills
(ECP\ref{item:ECP1}-\ref{item:ECP3}) in
Definition \ref{ecp}, then $Z(M_1)$ is non trivial and contains
${I_{\error}}$. Therefore, Step (\ref{item:stepECP_1}) of the
algorithm consists in computing $J = Z(M_1)$.

\subsubsection{Second step of the error correcting pair algorithm}
Step~(\ref{item:stepECP_2}) is nothing but the resolution of a linear
system depending on $J$ and the syndrome of $\yv$. First, some
notation is needed.

\begin{notat}
  Let $H$ be a matrix having $n$ columns and
  $J\subseteq\{1, \dots, n\}$. We denote by $H_J$ the submatrix of $H$
  whose columns are those with index $j\in J$.
\end{notat}

\medskip\noindent Suppose we have computed $J\supseteq {I_{\error}}$
in Step~(\ref{item:stepECP_1}) of the algorithm. Consider a full
rank--parity check matrix $H$ for $C$. The vector $\error_J$ satisfies
$H_J\cdot \error_J^T=H\cdot \yv^T.$
and we want then to recover $\error_J$ by solving the linear system 
\begin{equation}\label{systECP}
H_J\cdot \uv^T=H\cdot \yv^T.
\end{equation}
Though, \textit{a priori}, the solution may not be unique. Condition
(ECP\ref{item:ECP4}) in Definition \ref{ecp} yields the following result. 
\begin{lem}\label{lemsist}
If $\ \dd(A)+\dd(C)>n$, $\dim(A)>t$ and $J=Z(M_1)$, then $|J|<\dd(C)$.
\end{lem}
\begin{proof}
  By Proposition~\ref{dim(A)}, there exists
  $\av\in A({I_{\error}})\setminus \{0\}$. Now, since $\dd(A)+\dd(C)>n$,
  we get
\[|J|=|Z(M_1)|\le|Z(\av)|=n-\w(\av)\le n-\dd(A)< \dd(C).\]
\end{proof}
\begin{thm}\label{thm: thm standard}
  Given $\yv\in\mathbb{F}_q^n$ and $J\subseteq\{1,\dots, n\}$ with
  $t=|J|<\dd(C)$, then there exists at most one solution for
  (\ref{systECP}).
\end{thm}
\noindent
This is a well-known result of coding theory and it is easy to
prove. This theorem, together with Lemma~\ref{lemsist}, entail that if
$J$ contains the support of the error, $\error_J$ is the unique
solution to system (\ref{systECP}). Then, the second step of the
algorithm consists in finding $\error_J$ by solving system
(\ref{systECP}) and recovering $\error$ from $\error_J$ imposing
$e_i=0$ for all $i\notin J$. The entire algorithm is described in
Algorithm 3.
\begin{algorithm}\label{AlgoECP}
\caption{Error correcting pairs algorithm}
\textbf{Inputs:} $C$ linear code, $\yv\in\mathbb{F}_q^n$
as in (\ref{forme_yv}),$t=\w(\error)$, $(A, B)$ a $t$-error
correcting pair for $C$\\
\textbf{Output:} $\cv\in C$ such that $\yv=\cv+\error$ for some
$\error\in\mathbb{F}_q^n$ with $\w(\error)\le t$
\begin{algorithmic}[1]
  \State compute $M_1=\{\av\in A\mid \langle \av\ast \yv, \bv\rangle=0
  \ \ \forall \bv\in B\}$ (linear system)
\State $J\gets Z(M_1)$
\If {system (\ref{systECP}) does not have solution}
\State \Return failure
\EndIf
\State $\error_J\gets$ solution of (\ref{systECP})
\State recover $\error$ from $\error_J$
\State \Return $\cv=\yv-\error$
\end{algorithmic}
\end{algorithm}

\noindent The correctness of the algorithm is proved in \cite{P92}.   
 It is straightforward to see that the algorithm
returns a unique solution and that a sufficient condition for the
algorithm to correct $t$ errors, is the existence of an error
correcting pair with parameter $t$. This consideration, 
leads to the following result.
\begin{cor}[{\cite[Corollary 2.15]{P92}}]
If a linear code $C$ has a $t$-error correcting pair, then
\[t\le \Bigl\lfloor \frac{\dd(C)-1}{2}\Bigr\rfloor.\]
\end{cor}
\subsection{Error correcting pairs for Reed--Solomon
  codes}\label{secECPRS}
For an arbitrary code, there is no reason that an error correcting pair
exists. 
Indeed, the existence of an ECP for a
given code relies on the existence of a pair $(A, B)$ of codes, both
having a sufficiently large dimension and satisfying
$A * B \subseteq C^{\perp}$, which is actually a very restrictive
condition.
Among the codes for which an ECP exists, there are Reed--Solomon
codes. Indeed,
given $C=\RS{q}{k}$, consider the pair $(A, B)$:
\begin{equation}\label{ECPRS}
A=\RS{q}{t+1},\quad B^{\perp}=\RS{q}{t+k}.
\end{equation}
Recall that thanks to Proposition~\ref{prop:star_prod_GRS}, we have
$
  A * C = \RS{q}{t+k}.
$

\noindent

\begin{lem}\label{keylemECP}
Given $B$ as above, it holds
\begin{equation}\label{KeyeqECP}
\dd(B^{\perp})>t \iff t\le \frac{\dd(C)-1}{2}\cdot
\end{equation}
\end{lem}
\begin{prop}\label{PropisECP}
  The pair $(A, B)$ of (\ref{ECPRS}) is a $t$-error correcting pair
  for $C$ for any
\begin{equation}\label{halfmindist}
t\le \frac{\dd(C)-1}{2}\cdot
\end{equation}
\end{prop}
\begin{proof}
  We have to prove that (\ref{halfmindist}) is a necessary and
  sufficient condition for (ECP\ref{item:ECP1}--\ref{item:ECP4}) in
  Definition \ref{ecp} to hold. First of all, by Lemma~\ref{keylemECP}
  we have (ECP\ref{item:ECP3}). Moreover $\dim(A)=t+1>t$ by definition of
  $A$ and this gives (ECP\ref{item:ECP2}). By Proposition~\ref{prop:star_prod_GRS}, as seen above, the codes $A, B , C$
  verify $A\ast C=B^{\perp}$, then by Remark~\ref{propstarprod} we
  obtain $A\ast B\subseteq C^{\perp}$. Finally, it is easy to see that
  $\dd(A)+\dd(C)>n\iff t<\dd(C)$. Hence if
  $t\le\Bigl\lfloor \frac{\dd(C)-1}{2}\Bigr\rfloor$,
  (ECP\ref{item:ECP1}--\ref{item:ECP4}) hold and conversely.
\end{proof}

\begin{rem}\label{remELP}
  In \S~\ref{secPECPRS}, we are going to work with structures which
  are slightly different from error correcting pairs, that is, we will
  still require (ECP\ref{item:ECP1}, \ref{item:ECP2}) and
  (ECP\ref{item:ECP4}) to hold together with other
  conditions. 
  Note that, given $A$ and $B$ as in
  (\ref{ECPRS}) Conditions (ECP\ref{item:ECP1}, \ref{item:ECP2}) and
  (ECP\ref{item:ECP4})  hold if and only if
  $t<\dd(C)$. 
\end{rem}

\subsection{ECP and Welch--Berlekamp key equations}\label{subsec:ECPRS}
The example of Reed--Solomon also permits to understand the rationale
behind EPC's in light of Welch--Berlekamp algorithm.  Indeed, we now
show that the choice of $M_1$ we made in the ECP algorithm, if one looks
at the key equations of Welch--Berlekamp algorithm, appears to be really
natural. Let us consider $C=\RS{q}{k}$ and the pair $(A, B)$ we
defined in \S~\ref{secECPRS}. We can write (\ref{keyeqBW}) for any
$i\in\{1,\dots, n\}$ using the star product in this way
\[(\Lambda(x_1),\dots, \Lambda(x_n))\ast \yv=(N(x_1),\dots, N(x_n)).\]
From that, we can deduce
\begin{itemize}
\item $(N(x_1),\dots, N(x_n))\in \RS{q}{t+k}=B^{\perp}$;
\item
  $(\Lambda(x_1), \dots, \Lambda(x_n))\in
  \RS{q}{t+1}({I_{\error}})=A({I_{\error}})$;
\item Moreover
  $(\Lambda(x_1), \dots, \Lambda(x_n))\in \underbrace{\{\av\in A\mid
    \langle a\ast \yv, b\rangle=0\ \ \forall \bv\in B\}}_{M_1}$.
\end{itemize}
In other words, the vector $(\Lambda(x_1), \dots, \Lambda(x_n))$
belongs to the space $A({I_{\error}})$ we are looking for in the ECP
algorithm. Moreover it fulfills a property which characterises a space
$M_1\supseteq A({I_{\error}})$, that is exactly the space we define in
the ECP algorithm and that turns to be equal to $A({I_{\error}})$
under certain conditions.

 \section{Power error locating pairs algorithm}\label{secPECPRS}
We now present the Power Error Locating Pairs (PELP) algorithm. As for
the error correcting pairs algorithm, we first give a generic
description of the algorithm and later some examples of its
application. In order to generalise the ECP algorithm to correct more
errors, we introduce a new parameter $\ell$ we call \textit{power} and
define a slightly different structure from error correcting pairs. As
in the previous paragraphs, we first describe that structure and the
algorithm for $\ell=2$ and then explain how to generalise it.
\subsection{The case $\ell=2$}
 In Pellikaan's paper, a structure called {\em error locating
  pairs} is already defined. It is a pair of codes $(A, B)$ which
satisfy (ECP\ref{item:pelp1}, \ref{item:pelp2}) and
(ECP\ref{item:pelp4}). In particular it is shown that, without changing
anything in the algorithm, with such a structure it is possible to
correct errors if the support of
the error vector ${I_{\error}}$ is an \textit{independent $t$-set of
  error position} with respect to the code $B$ (see \cite{P92}).

In the present article, in order to correct beyond half the
designed distance, we do not consider particular error supports, but we
rather choose to work with a more particular structure than error
locating pairs and change the first step of the algorithm.

\begin{defn}[$2$--Power error locating pairs]\label{def:PELP2}
  Given a linear code $C\subseteq\mathbb{F}_q^n$, a pair of linear
  codes $(A, B)$ with $A, B\subseteq\mathbb{F}_q^n$ is a
  $2$--\textit{power} $t$--\textit{error locating pair} for $C$ if
\begin{enumerate}[($2$--PELP1)]
\item\label{item:pelp1} $A\ast B\subseteq C^\perp$;
\item\label{item:pelp2} $\dim(A)>t$;
\item\label{item:pelp3} $\dd(A^{\perp})>t$;
\item\label{item:pelp4} $\dd(A)+\dd(C)>n$;
\item\label{item:pelp5} $\dim(B)+\dim(B^{\perp}\ast C)^{\perp}\ge t$.
\end{enumerate}
\end{defn}

Compared to the definition of error correcting pairs, we removed
(ECP~\ref{item:ECP3}) which is too restrictive to correct errors
beyond half the minimum distance. Instead, in the same spirit as the
power decoding, we look for a necessary condition for the algorithm to
succeed. In this context, under Condition~($2$--PELP\ref{item:pelp3}),
Condition~($2$--PELP\ref{item:pelp5}) provides this necessary condition
together with the key tool for the analysis of the decoding radius of
our algorithm.
\begin{rem}\label{rem_propB}
  In the transition between ECP algorithm and $2$--PELP algorithm, it is
  very important to get rid of the property $\dd(B^{\perp})>t$. Indeed,
  since we want $A\ast C\subseteq B^{\perp}$, if we had
  $\dd(B^{\perp})>t$, then, assuming that $A * C$ is non degenerate (see
  \S~\ref{sss:Kneser})
  we would get
  \[t+\dim(C)\le\dim(A)+\dim(C)-1\le \dim(A\ast C)\le
    \dim(B^{\perp})\le n-t, \] where the second inequality is
  due to Corollary~\ref{cor:CauchyDavenport}.
  This entails that
  $t\le \frac{n-\dim(C)}{2}$, which does not represent an improvement
  of the decoding radius in every situation (see
  \S\ref{sec:PELP_RS}).
\end{rem}

\medskip

Let us consider a code $C$, a word $\yv\in\mathbb{F}_q^n$ such that
$\yv=\cv +\error$ with $\cv\in C$ and $t=\w(\error)$ and let $(A, B)$
be a $2$--power $t$--error locating pair for $C$.

\begin{defn}\label{def:M-M_1}
As in the ECP algorithm, let
  $M_1=\{\av\in A\mid \langle \av\ast\yv, \bv\rangle=0\ \ \forall
  \bv\in B\}$. We then define the spaces
\begin{align*}
  M_2 &\eqdef\{\av\in A\mid \langle \av\ast\yv^2, \vv\rangle=0\ \ \forall \vv\in (B^{\perp}\ast C)^{\perp}\}
\end{align*}
and \[M\eqdef M_1\cap\ M_2.\]
\end{defn}

\begin{notat}
  As in the description of the power decoding algorithm, since we work
  with $\yv$ and $\yv^2$, we indicate with $\error^{(1)}$ and
  $\error^{(2)}$ the vectors such that respectively
  $\yv=\cv+\error^{(1)}$ and $\yv^2=\cv^2+\error^{(2)}$.
\end{notat}

\begin{algorithm}
  \caption{$2$--Power error locating pairs algorithm}
  \label{algoPELP2}
\textbf{Inputs:} $C$ linear code, $\yv\in\mathbb{F}_q^n$ as
in (\ref{forme_yv}), $t=\w(\error)$ and  $(A, B)$ a $2$--power $t$--error
locating pair for $C$\\
\textbf{Output:} $\cv\in C$ such that $\yv=\cv+\error^{(1)}$ for some
$\error^{(1)}\in\mathbb{F}_q^n$ with $\w(\error^{(1)})\le t$
\begin{algorithmic}[1]
\State compute $M=M_1\cap M_2$ by solving a linear system
\State $J\gets Z(M)$
\If {system (\ref{systECP}) does not have any nonzero solution}
\State \Return failure
\EndIf 
\State $\error^{(1)}_J\gets$ nonzero solution of (\ref{systECP})
\State recover $\error^{(1)}$ from $\error^{(1)}_J$
\State \Return $\cv=\yv-\error^{(1)}$
\end{algorithmic}
\end{algorithm}

\medskip\noindent As shown in Algorithm~\ref{algoPELP2}, the only
change from the error correcting pair algorithm consists in computing
a new set $M$, which is smaller than the set $M_1$ considered in the
basic algorithm. The reason why we do so, is that since we are now
working with a $2$--power error locating pair, we are no longer asking
for $\dd(B^{\perp})>t$. Hence, without this property the equality
\begin{equation}\nonumber
A({I_{\error}})=M_1
\end{equation} 
is not entailed and only the inclusion $A({I_{\error}})\subseteq M_1$
still holds.  That is, $M_1$ could be too big with respect to
$A({I_{\error}})$ and the algorithm would
fail. Indeed, if there was $\av\in M_1\setminus A({I_{\error}})$,
that means
$\av_{I_{\error}}\ne 0$, we would have
$J=Z(M_1)\nsupseteq {I_{\error}}$.
\begin{prop}\label{conten}
The set $M$ fulfills $A({I_{\error}})\subseteq M\subseteq M_1\subseteq A$.
\end{prop}
\begin{proof}
  We have to prove that $A({I_{\error}})\subseteq M$, that is
  $A({I_{\error}})\subseteq M_1, M_2$. First, we already know that
  $A({I_{\error}})\subseteq M_1$ by Theorem~\ref{thmecp}(1). About
  $M_2$, we adapt the proof of Theorem~\ref{thmecp}(1). First, note
  that we have
\begin{equation}\label{extPowstar}
A\ast B\subseteq C^{\perp}\iff A\ast C\subseteq B^{\perp}\implies A\ast C^2\subseteq B^{\perp}\ast C.
\end{equation}
Given $\av\in A({I_{\error}})$, for any
$\vv\in(B^{\perp}\ast C)^{\perp}$ we obtain
\begin{eqnarray}
\langle \av\ast \yv^2, \vv\rangle &= &\langle \av\ast \cv^2, \vv\rangle+\langle \av\ast \error^{(2)}, \vv\rangle \label{premPow}\\
						  &= & 0.\label{troisiPow}
\end{eqnarray}  
In (\ref{premPow}) we used the bilinearity, while (\ref{troisiPow})
holds because of (\ref{extPowstar}) and the fact that the supports of
$\error^{(2)}$ and $\av$ are complementary (see
Proposition~\ref{remsupp}).
\end{proof}

\medskip
\noindent
Thanks to this proposition, one can see that if we work with a $2$--PELP,
there are more chances to have $A({I_{\error}})=M$ than
$A({I_{\error}})=M_1$, that is why in $2$--PELP algorithm, we look for $M$
instead of $M_1$.  Furthermore, we get the following result.
\begin{thm}\label{corrPELP}
Under Assumption~\ref{as:a_solution_exists}, Algorithm~\ref{algoPELP2}
returns $\cv$ if and only if $A({I_{\error}}) = M$.
\end{thm}
\begin{proof}
  Suppose the algorithm returned $\cv$. This entails that we computed
  a set $J = Z(M)$ which contains $I_{\error}$ and hence $M =
  M(I_{\error})$. Therefore, from Proposition~\ref{conten}, we get
  $M = A(I_{{\error}})$.  Conversely, if $M = A(I_{\error})$, then
  $J = Z(M) \supseteq I_{\error}$. Next, using Conditions
  ($2$--PELP\ref{item:pelp2}) and ($2$--PELP\ref{item:pelp4}), one can prove a result
  very similar to Lemma~\ref{lemsist} (replacing $M_1$ by $M$)
  asserting that the linear system
  (\ref{systECP}) has a unique solution $\error$. Thus, the algorithm
  returns $\cv$.
\end{proof}

\medskip
\noindent
We now show how the properties ($2$--PELP\ref{item:pelp3}) and
($2$--PELP\ref{item:pelp5}) are involved in $2$--PELP algorithm. To do so, we want
to focus on necessary conditions for the algorithm to return $\cv$. By Theorem
\ref{corrPELP}, this is equivalent to look for a necessary condition
to have $A({I_{\error}})=M$, which is the point of the following statement
and explains the rationale behind Condition ($2$--PELP\ref{item:pelp5}).
\begin{thm}\label{thm:nec_cond}
If $A({I_{\error}})=M$, then $\dim(B)+\dim(B^{\perp}\ast C)^{\perp}\ge t$.
\end{thm}

\medskip
\noindent Before proving Theorem~\ref{thm:nec_cond}, we need to
present some other results first.
\begin{lem}\label{cond_eq_M}
The following statements are equivalent:
\begin{enumerate}[(i)]
\item $M=A({I_{\error}})$;
\item $M({I_{\error}})=M$;
\item $M_{I_{\error}}=\{0\}$.
\end{enumerate}
\end{lem}
\begin{proof}
  First, notice that $A({I_{\error}})=M({I_{\error}})$ by
  Proposition~\ref{conten}. Hence, \textit{(i)} and \textit{(ii)} are
  equivalent. We now prove \textit{(ii)}$\iff$\textit{(iii)}. It is
  easy to see that, if $M({I_{\error}})=M$, the projection on
  ${I_{\error}}$ of every element of $M$, is the zero vector and
  conversely.
\end{proof}

\medskip\noindent Thanks to this lemma, if we find a necessary
condition for $M_{I_{\error}}=\{0\}$, then we find a necessary
condition for $M=A({I_{\error}})$ and conversely. For this reason, we
now study the object~$M_{I_{\error}}$.
\begin{thm}\label{scritM_I}
  It holds
  $M_{I_{\error}}=A_{I_{\error}}\cap (\error^{(1)}\ast
  B)_{I_{\error}}^{\perp}\cap {\left(\error^{(2)}\ast (B^{\perp}\ast
  C)^{\perp}\right)}_{I_{\error}}^{\perp}$.
\end{thm}

\begin{rem}
  Note that the notation $(\error^{(1)}\ast B)_{I_{\error}}^{\perp}$
  and
  ${\left(\error^{(2)}\ast (B^{\perp}\ast
      C)^{\perp}\right)}_{I_{\error}}^{\perp}$ could be considered as
  ambiguous since the operation of puncturing and that consisting in
  taking the dual code do not commute. Actually, in this situation,
  there is no ambiguity since the supports of $\error^{(1)}\ast B$ and
  $\error^{(2)}\ast (B^{\perp}\ast C)^\perp$ are contained in
  $I_{\error}$.
\end{rem}

\begin{proof}[Proof of Theorem~\ref{scritM_I}]
  Let us characterise the elements of $A$ that
  are in $M_1$ and $M_2$. Let $\av\in A$,
\begin{eqnarray*}
  \av\in M_1 & \iff & \langle \av\ast \yv, \bv\rangle=0\quad \forall \bv\in B\\
             & \iff & \langle \av, \error^{(1)}\ast \bv\rangle=0\quad \forall \bv\in B\\
             & \iff & \av_{I_{\error}}\in (\error^{(1)}\ast B)_{I_{\error}}^{\perp}.
\end{eqnarray*}
For the second equivalence we used the property
$A\ast B\subseteq C^{\perp}$ which is equivalent to
$A\ast C\subseteq B^{\perp}$ (see Remark~\ref{rem: swich}) while in
the third, we used $\supp(\error^{(1)})={I_{\error}}$. In the same way, it is
possible to prove that given $\av\in A$,
\begin{equation}\nonumber
  \av\in M_2\iff \av_{I_{\error}}\in (\error^{(2)}\ast (B^{\perp}\ast C)^{\perp})_{I_{\error}}^{\perp}.
\end{equation}
The result comes now easily.
\end{proof}

\medskip\noindent
It is actually possible to simplify the form of $M_{I_{\error}}$.
\begin{lem}\label{lemma:equivA_I=0}
  Let $A\subseteq \mathbb{F}_q^n$ be a linear code and
  $t\in\mathbb{N}$. The following facts are equivalent:
\begin{itemize}
\item $\dd(A^\perp)>t$;
\item $A_{J}=\mathbb{F}_q^t$ for any ${J}\subseteq\{1,\dots, n\}$ with
  $|{J}|=t$.
\end{itemize}
\end{lem}

\begin{cor}
If $(A, B)$ is a $2$--power $t$--error locating pair for the code $C$, then 
\[M_{I_{\error}}=(\error^{(1)}\ast B)_{I_{\error}}^{\perp}\cap
  (\error^{(2)}\ast (B^{\perp}\ast C)^{\perp})_{I_{\error}}^{\perp}.\]
\end{cor}

\medskip\noindent It is now possible to prove Theorem \ref{thm:nec_cond}.
\begin{proof}[Proof Theorem \ref{thm:nec_cond}]
Looking for a necessary condition to have $M_{I_{\error}}=\{0\}$, we get
\begin{eqnarray}
M_{I_{\error}}=\{0\} & \iff & (\error^{(1)}\ast B)_{I_{\error}}^{\perp}\cap (\error^{(2)}\ast (B^{\perp}\ast C)^{\perp})_{I_{\error}}^{\perp}=\{0\} \\
                     &\implies & \dim((\error^{(1)}\ast B)_{I_{\error}}^{\perp}) + \dim((\error^{(2)}\ast (B^{\perp}\ast C)^{\perp})_{I_{\error}}^{\perp})\le t \label{neccond2}\\
                     & \implies & \dim (\error^{(1)}\ast B)_{I_{\error}} +
                                  \dim (\error^{(2)}\ast (B^{\perp}\ast C)^{\perp})_{I_{\error}} \ge t
  \\
          & \implies & \dim(B)+\dim((B^{\perp}\ast C)^{\perp})\ge t.
                       \label{eq:cond_for_PELP}
\end{eqnarray}
\end{proof}

\subsection{The general case: $\ell\ge 2$}
It is possible to generalise the algorithm for $\ell\ge 2$. First we
generalise the structure we use.
\begin{defn}[$\ell$--Power error locating pairs]
  Given a linear code $C\subseteq\mathbb{F}_q^n$, a pair of linear
  codes $(A, B)$ with $A, B\subseteq\mathbb{F}_q^n$ is an
  $\ell$--\textit{power} $t$--\textit{error locating pair} for $C$ if
\begin{enumerate}[($\ell$--PELP1)]
\item\label{item:pelp1g} $A\ast B\subseteq C^\perp$;
\item\label{item:pelp2g} $\dim(A)>t$;
\item\label{item:pelp3g} $\dd(A^{\perp})>t$;
\item\label{item:pelp4g} $\dd(A)+\dd(C)>n$;
\item\label{item:pelp5g}
  $\dim(B)+\sum_{i=2}^{\ell}\dim(B^{\perp}\ast C^{i-1})^{\perp}\ge t$.
\end{enumerate}
\end{defn}

\medskip\noindent A generalisation of the notion of $M$ is needed too,
that is, for a generic power $\ell$, $M$ will be the intersection of
$\ell$ spaces.
\begin{defn}
  As in the ECP algorithm, let
  $M_1=\{\av\in A\mid \langle \av\ast\yv, \bv\rangle=0\ \ \forall
  \bv\in B\}$. For any $i=2,\dots,\ell$, we define the space
\begin{align*}
  M_i &\eqdef\{\av\in A\mid \langle \av\ast\yv^i, \vv\rangle=0\ \ \forall \vv\in (B^{\perp}\ast C^{i-1})^{\perp}\}.
\end{align*}
Finally we define \begin{equation}\label{eq:newM}
  M\eqdef \bigcap_{i=1}^{\ell}M_i.
\end{equation}
\end{defn}

\medskip\noindent Then, the algorithm for a general $\ell$ is the same
as Algorithm~\ref{algoPELP2} changing only the definition of $M$ by that
given in \eqref{eq:newM}.
Let us look for a necessary condition for this generalised algorithm
to return $\cv$.  It can be proven that Theorem \ref{corrPELP} can be
adapted to the generalised notion \eqref{eq:newM} of $M$. The
following theorem gives the necessary condition we look for.
\begin{thm}\label{corr_PELPg}
  If $A({I_{\error}})=M$, then
  $\dim(B)+\sum_{i=2}^{\ell}\dim(B^{\perp}\ast C^{i-1})^{\perp}\ge t.$
\end{thm}
\medskip\noindent Again, in order to prove this Theorem, we study the
condition $M_{I_{\error}}=\{0\}$, since it is equivalent to
$A({I_{\error}})=M$ by Lemma \ref{cond_eq_M}.
\begin{thm}
It holds 
\begin{equation}\nonumber
  M_{I_{\error}}=(\error^{(1)}\ast B)_{I_{\error}}^{\perp}\cap \bigcap_{i=2}^{\ell}
  {\left(\error^{(i)}\ast (B^{\perp}\ast C^{i-1})^{\perp}\right)}_{I_{\error}}^{\perp},
\end{equation}
where $\error^{(i)}=\sum_{h=1}^i \binom{i}{h}\cv^{i-h}\ast \error^{h}$
is such that $\yv^i=\cv^i+\error^{(i)}$ for all $i=1,\dots, \ell$.
\end{thm} 

\medskip\noindent To prove this result, it is possible to adapt the
proof of Theorem \ref{scritM_I} and observe that it still holds
$A_{I_{\error}}=\mathbb{F}_q^t$. We will use the following remark.
\begin{rem}\label{remboundellg} $ $
  Given a vector space $V$ with $\dim(V)=t$ and
  $A_1,\dots, A_n\subseteq V$, we have
\[\dim\Big(\bigcap_{i=1}^n A_i^{\perp}\Big)\ge t-\sum\dim(A_i);\]
in addition, it is easy to see that
  $\dim((\error^{(1)}\ast B)_{I_{\error}})\le \dim(B)$ and
\begin{equation}\nonumber
  \forall i\in \{2,\dots, \ell\}, \quad  \dim((\error^{(i)}\ast (B^{\perp}\ast C^{i-1})^{\perp})_{I_{\error}})\le \dim((B^{\perp}\ast C^{i-1})^{\perp}).
\end{equation}
\end{rem}

\medskip\noindent Now, it is possible to prove Theorem \ref{corr_PELPg}.
\begin{proof}[Proof Theorem \ref{corr_PELPg}]
It holds
\begin{eqnarray}
  M_{I_{\error}}=\{0\} & \iff & (\error^{(1)}\ast B)_{I_{\error}}^{\perp}\cap
                                \bigcap_{i=2}^{\ell}{\left(\error^{(i)}\ast (B^{\perp}\ast C^{i-1})^{\perp}\right)}_{I_{\error}}^{\perp}=\{0\}\nonumber \\
                       &\implies & \dim((\error^{(1)}\ast B)_{I_{\error}}^{\perp}) + \dim\Big(\bigcap_{i=2}^{\ell}(\error^{(i)}\ast (B^{\perp}\ast C^{i-1})^{\perp})_{I_{\error}}^{\perp}\Big)\le t. \label{neccond2g}
\end{eqnarray}
Now, thanks to Remark \ref{remboundellg}, one can easily see that
(\ref{neccond2g}) implies
\[\dim(B)+\sum_{i=2}^{\ell}\dim(B^{\perp}\ast C^{i-1})^{\perp}\ge t.\]
\end{proof}

\subsection{Complexity}
To conclude this section, let us discuss the complexity of the
algorithm. We denote by $\omega$ the exponent of the complexity of
matrix multiplications. First, recall that the computation of the
star product of two codes of length $n$ costs $O(n^4)$ arithmetic operations in
$\Fq$ using a deterministic algorithm and $O(n^\omega)$ using a
probabilistic algorithm (see for instance \cite[\S~VI.A and
D]{COT17}).

The evaluation of the complexity of the power error locating pairs
algorithm should be divided in two parts:
\begin{itemize}
\item the {\em pre-computation phase}, that should be done once for good and
  is independent from the error and the corrupted codeword;
\item the {\em online phase}, which depends on the corrupted codeword.
\end{itemize}

\subsubsection{The precomputation phase}
This phase consists essentially in computing generator matrices for
the codes ${(B^\perp*C^{i-1})}^\perp$ for $i \in \{1, \dots, \ell\}$.
Each new calculation consists in the computation of a $*$--product and
a dual.  This yields an overall cost of $O(\ell n^\omega)$ operations
in $\Fq$ using a probabilistic algorithm and $O(\ell n^4)$ operations
using a deterministic one.

\subsubsection{The online phase}

\begin{itemize}
\item The computation of each space $M_i$ boils down to the resolution
  of a linear system with $\dim A$ variables and
  $\dim {(B^\perp * C^{i-1})}^\perp$ equations.
  Hence a cost of $O(n^\omega)$ operations in $\Fq$.
\item The computation of $M$ consists in the calculation of $\ell - 1$
  intersections of spaces. Since the cost of the calculation of an
  intersection is $O(n^\omega)$ operations, the cost of the
  computation of $M$ from the knowledge of the $M_i$'s is in
  $O(\ell n^\omega)$
\end{itemize}
In summary, the overall complexity of the online phase is in $O(\ell n^\omega)$
operations in $\Fq$.

\begin{rem}
  Note that the previous complexity analysis is purely generic and
  does not take into account that codes with an error locating pair
  such as Reed--Solomon code may be described by structured matrices
  permitting faster linear algebra.
\end{rem}

\section{$\ell$--PELP algorithm for Reed--Solomon codes}\label{sec:PELP_RS}
We now give some applications of the $\ell$--PELP algorithm, starting with
Reed--Solomon codes. For these codes, the algorithm is much more
intuitive. Indeed, as for the error correcting pairs algorithm, it is
possible to deduce the PELP algorithm from a former decoding algorithm
for Reed--Solomon codes: the {\em power decoding}.

Let us consider the code $C=\RS{q}{k}$ and the pair $(A,B)$, where
$A=\RS{q}{t+1}$ and $B^{\perp}=\RS{q}{t+k}$. We look for the values of
$t$ for which $(A, B)$ is an $\ell$--power $t$--error locating pair for $C$. One can
see that, since we no longer ask for $\dd(B^{\perp})>t$, by Lemma
\ref{keylemECP}, $t$ can be larger than $\frac{\dd(C)-1}{2}$. About
the conditions to fulfill, we already have seen that properties
($\ell$--PELP\ref{item:pelp1}, \ref{item:pelp2}, \ref{item:pelp4})
hold for any $t<\dd(C)$ (see \S\ref{secECPRS}). Let us find the values
of $t$ which verify
\begin{description}
\item[($\ell$--PELP\ref{item:pelp3g})] $\dd(A^{\perp})>t$;
\item[($\ell$--PELP\ref{item:pelp5g})] $\dim(B)+ \sum_{i=2}^\ell \dim(B^{\perp}\ast C^{i-1})^{\perp}\ge t$.
\end{description}
Property ($\ell$--PELP\ref{item:pelp3}) holds for any $t$ since Reed--Solomon
codes are MDS and $A$ has dimension $t+1$. Let us now focus on
property ($\ell$--PELP\ref{item:pelp5}).
By Proposition \ref{prop:star_prod_GRS}, we know that
$B^{\perp}\ast C^{i-1}=\RS{q}{t+i(k-1)+1}$ and these codes are not equal
to $\Fq^n$ as soon as
\begin{equation}\label{weak_cond2}
t< n- \ell (k-1) - 1.
\end{equation}
\medskip
\noindent
If (\ref{weak_cond2}) is satisfied, then the bound in property
($\ell$--PELP\ref{item:pelp5}) becomes
\begin{equation}\label{dec_rad:PELP2RS}
t\le\frac{2n\ell-k\ell(\ell+1)+\ell(\ell-1)}{2(\ell+1)},
\end{equation}
which is the decoding radius for the power decoding algorithm for
a general $\ell$ (see \eqref{decradPow}).
\begin{rem}
Note that (\ref{dec_rad:PELP2RS})
  came in power decoding as a necessary condition to have a unique
  solution for a linear system. Here instead, it comes up as a
  necessary condition for an intersection of some vector spaces to
  be $\{0\}$.
\end{rem}

\subsection{The space $M$ and the key equations of power decoding}
In \S~\ref{subsec:ECPRS}, we have seen that it is possible to relate the
definition of $M_1$ with the key equations of Welch--Berlekamp algorithm.
One can do the same with the definition of $M$ in the power error
locating pairs algorithm and the key equations of the $\ell$--power decoding
algorithm. Here, we only consider the case $\ell=2$, since it is easy to
generalise the idea for a larger $\ell$. It is possible to write
(\ref{Keyeqpow}) in this way
\[
  \left\{
    \begin{array}{lcl}
      \ev_{\xv}(\Lambda)\ast \yv &=& \ev_{\xv}(N_1)\\
      \ev_{\xv}(\Lambda)\ast \yv^2&=& \ev_{\xv} (N_2),
    \end{array}
  \right.
\]
where $\ev_{\xv}$ is the evaluation map introduced in
\eqref{eq:ev_map}.  Hence, we can
deduce 
\begin{itemize}
\item $\ev_{\xv}(N_1)\in \RS{q}{t+k}=B^{\perp}$;
\item $\ev_{\xv}(N_2)\in \RS{q}{t+2k-1}=B^{\perp}\ast C$;
\item $\ev_{\xv}(\Lambda)\in \RS{q}{t+1}({I_{\error}})=A({I_{\error}})$;
\item $\ev_{\xv}(\Lambda)\in M$, where $M$ is the set defined in the
  $2$--power error locating pairs algorithm. Indeed we recall that
  $M=M_1\cap M_2$, where
\begin{eqnarray*}
  M_1 &= &\{\av\in A\mid \langle \av\ast \yv, \bv\rangle=0\ \forall \bv\in B\},\\
  M_2 &= &\{\av\in A\mid \langle \av\ast \yv^2, \vv\rangle=0\ \forall \vv\in (B^{\perp}\ast C)^{\perp}\}.
\end{eqnarray*}
\end{itemize}

\medskip\noindent In other words, in the power decoding algorithm one
works with polynomials, while in the power error locating pairs
algorithm one works with their evaluations. 

\subsection{Equivalence of the two algorithms for Reed--Solomon codes}
Thanks to the link presented in the previous subsection, it is
possible to find an isomorphism between the solution space of power
decoding and the space $M$. For the sake of simplicity, we explicit
this isomorphism in the case $\ell = 2$. The general case can easily
be deduced from the following study.
\begin{thm}\label{thm:iso_sol_m}
  Let $\yv=\mathbb{F}_q^n$ and $t$ a positive integer and suppose we
  run both the power decoding algorithm and the power error locating
  pairs algorithm with the same $t$ and $\ell=2$. Denote by $Sol$ the solution
  space of the linear system (\ref{eq:Spo}) in the power
  decoding. Then the linear map
\begin{equation}\label{eq:phi}
  \varphi \eqdef
  \map{Sol}{M}{(\lambda, \nu_1, \nu_2)}{\ev_{\xv}(\lambda).}
\end{equation}
is bijective.
\end{thm}
\begin{proof}
  For the sake of simplicity, we provide the proof in the case
  $\ell = 2$.  The proof in the general case is easy to deduce at the
  cost of heavier notation. First, let us show that $\varphi$ is well
  defined. Let $(\lambda, \nu_1, \nu_2)$ belong to $Sol$. Then, it
  holds
\begin{equation}\label{keqfail}
\begin{cases}
\lambda(x_i)y_i=\nu_1(x_i) & i=1,\dots, n\\
\lambda(x_i)y_i^2=\nu_2(x_i) & i=1,\dots, n.
\end{cases}
\end{equation}
As we have seen in \S~\ref{sec:PELP_RS}, these two
conditions are equivalent to the statement
\begin{equation}\nonumber
\ev_{\xv}(\lambda)\in M_1\cap M_2=M.
\end{equation}
Conversely, given $\av\in M$, there exists
$\lambda\in \mathbb{F}_q[X]$ with $\deg(\lambda)<t+1$ such that
$\ev_{\xv}(\lambda)=\av$. Moreover, since $\av\in M$, we have
\begin{align*}
\av\ast \yv\in B^{\perp}  =\RS{q}{t+k},\ \ \ 
\av\ast\yv^2\in B^{\perp}\ast C  =\RS{q}{t+2k-1}.
\end{align*}
Thus, there exist $\nu_1, \nu_2\in\mathbb{F}_q[X]$ with
$\deg(\nu_1)<t+k$, $\deg(\nu_2)<t+2k-1$ and
\begin{equation}
\ev_{\xv}(\nu_1) =\av\ast \yv,\ \ \ \ev_{\xv}(\nu_2) =\av\ast\yv^2.
\end{equation}
We can then define another map
\begin{equation}
\psi\eqdef \map{M}{Sol}{\av}{(\lambda, \nu_1, \nu_2)},
\end{equation}
where $\lambda, \nu_1$ and $\nu_2$ are the polynomials associated to
$\av$ as before. It is easy to prove that, under the
condition\footnote{That is again bound (\ref{nec_cond: Pow_0}). }
\begin{equation}\label{eq:reminder}
t<n-2(k-1)
\end{equation}
it holds $\varphi\circ \psi=Id_{M}$ and $\psi\circ\varphi=Id_{Sol}$.
\end{proof}

In summary, for Reed--Solomon codes, power decoding and power error
locating pairs algorithms are equivalent. In particular they succeed
or fail for the same instances.

\section{PELP algorithm for algebraic
  geometry codes}\label{section:AGcodes}
As said previously, the power error locating pairs algorithm can be
run on any code with a PELP. We have seen that Reed--Solomon codes
belong to this class of codes. In the sequel, we show that algebraic
geometry codes also belong to it.  Similarly to the case of
Reed--Solomon codes, this algorithm can be compared with the power
decoding algorithm. Power decoding extends naturally from
Reed--Solomon codes to algebraic geometry codes. However, its use for
decoding AG codes in the literature concerns mainly one--point codes
from the Hermitian curve (see \cite{NB15,PRB19}). For this reason, we
give a brief presentation together with an analysis of its decoding
radius in Appendix~\ref{sec:appendix}.

In the sequel, we show that the analysis of the power decoding
provides a decoding radius which is slightly below that of the power
error locating pairs algorithm. Moreover, we observed experimentally
that the decoding radius given by the analysis of the PELP algorithm
is optimal for both the PELP and the power decoding
algorithms. Probably, a more detailed analysis of the power decoding
would provide a sharper estimate of the decoding radius, but the point
is that the analysis of the PELP algorithm provides an optimal radius
in a very elementary manner.

\subsection{Context}
Let $\X$ be a smooth projective geometrically connected curve of genus
$g$ over $\Fq$. Let $G$ be a divisor on $\X$ and
$\points = (P_1, \dots, P_n)$ be an ordered $n$--tuple of pairwise distinct rational
points of $\X$ avoiding the support of $G$.  We denote by $k$ and $\dd$ respectively the dimension and
the minimum distance of the code $\AGcode{\X}{\points}{G}$.  Moreover,
we denote by $D$ the divisor $P_1+\dots +P_n$ and by $W$ the divisor
$(\omega)$ where $\omega\in\Omega(\X)$ is a rational differential form
such that $v_{P_i} (\omega)=~-1$ and $res_{P_i}(\omega)=1$ for any
$i \in \{1,\dots, n\}$. 
We now introduce an extra divisor $F$ on $\X$ and the pair $(A, B)$ with
\begin{equation}\label{eq:elp_for_ag}
A=C_L(\X, \mathcal{P}, F)\ \ \ B=C_L(\X,\mathcal{P}, D+W-F-G).
\end{equation}
This pair of codes is our candidate to be a power error locating pair
for $C$. We analyse the case $\ell=2$ for simplicity (it is easy to
generalise what we are going to see).

\subsection{Decoding Radius}
In order to find the decoding radius of the $2$--power error correcting
pairs algorithm for algebraic geometry codes, we follow the same path
as for Reed--Solomon codes. That is, we look for the pairs $(A,B)$ that satisfy properties ($2$--PELP\ref{item:pelp1}--\ref{item:pelp5}) in Definition \ref{def:PELP2}.
To do so, we write some additional conditions on the degree of the
divisor $F$ and $G$ and on the number of errors $t$. First, note that
Property ($2$--PELP\ref{item:pelp1}) holds by construction of $A$ and $B$. In
order to have properties ($2$--PELP\ref{item:pelp2}, \ref{item:pelp3}) and to know the structure of the code $(B^{\perp}\ast C)^{\perp}$, we
ask for the two following conditions

\begin{cond}\label{condegF}
$\deg(F)\ge t+2g$.
\end{cond}
\begin{cond}
$\deg(G)\ge 2g$.
\end{cond}
\medskip\noindent In particular, it is easy to verify that under these two
additional conditions, we have by
Proposition~\ref{prop:star_prod_AGcodes}
\[(B^{\perp}\ast C)^{\perp}=C_L(\X, \mathcal{P},D+W-F-2G).\]
Let us fix then the value of $\deg(F)$ to be $2g+t$. We now consider the bound given by Condition ($2$--PELP\ref{item:pelp5}) for $\ell=2$
\begin{equation}\label{eq:AG_necessary}
\dim(B)+\dim((B^{\perp}\ast C)^{\perp})\ge t.
\end{equation}
We need to know the exact dimension of these spaces, hence we impose some more conditions on the
degree of the divisor\footnote{Remember that if $C=C_L(\X, \mathcal{P}, G)$ with $2g-2< \deg(G)<n$, then $\dim(C)=\deg(G)-g+1$.} $G$. We ask for 
\begin{cond}
$t<n-2\deg(G)-2g$.
\end{cond}
\noindent Finally, we get the following result.
\begin{prop}\label{thm:pelp_AG}
  Let $\deg(F)=t+2g<n-\deg(G)$, $\deg(G)< 2g$ and
  $t\le n-2\deg(G)-2g$. Then $C=C_L(\X, \mathcal{P}, G)$
  admits a $2$--PELP as in (\ref{eq:elp_for_ag}), if
\begin{equation}\label{rdPECPAG}
t\le \frac{2n-3\deg(G)-2}{3}-\frac{2}{3}g.
\end{equation}
In this case, bound (\ref{rdPECPAG}) gives the decoding radius of the $2$--PELP algorithm.
\end{prop}

\begin{proof}
  Condition ($2$--PELP\ref{item:pelp1}) is obviously satisfied by the codes
  $A, B$ defined in~(\ref{eq:elp_for_ag}). Moreover, since
  $\deg (F) = t+2g$, we get $\dim(A) > t$,
  i.e. Property~($2$--PELP\ref{item:pelp2}). Property~($2$--PELP\ref{item:pelp4})
  is a consequence of the condition $\deg (F) < n - \deg(G)$, which
  indeed entails \[\dd(A)+\dd(C) \geq 2n - \deg (F) - \deg (G) > n.\]
  Thanks to Additional Condition 1, we have ($2$--PELP\ref{item:pelp3}). Finally one notes that bound in ($2$--PELP\ref{item:pelp5}) becomes the
  bound on $t$ in (\ref{rdPECPAG}) thanks to the additional conditions
  and the property $\deg(F)<n-\deg(G)$.
\end{proof}
\begin{rem}
As for Reed--Solomon codes, we want to have $(B^{\perp}\ast C)\subsetneq \mathbb{F}_q^n$. Indeed in this case if $M$ and $M_1$ are as in Definition \ref{def:M-M_1}, we get $M\subsetneq M_1$ and the decoding radius in (\ref{rdPECPAG}) is usually achieved according to our tests. That is why it is important also to ask
\[t<n-2\deg(G)-g-1.\]
Note that this bound is achieved whenever we are in the hypothesis of Proposition \ref{thm:pelp_AG} and $g>1$.
\end{rem}
\medskip\noindent
The decoding radius can be computed even for
arbitrary values of $\ell$. Indeed, if we impose
$ t\le n-\ell\deg(G)-2g$, we get
\begin{equation}\label{Dec_rad:PELP_AG}
t\le \frac{2n\ell -\ell(\ell+1)\deg(G)}{2(\ell+1)}-g+\frac{g-\ell}{\ell+1}\cdot
\end{equation}

\subsection{Comparison with decoding radii of other algorithms for
  algebraic geometry codes}
We can now compare this decoding radius with the decoding radii of
Sudan algorithm and the power decoding algorithm for algebraic
geometry codes. We have (see \cite[Theorem 2.1]{SW99} and
Appendix~\ref{sec:appendix}):
\begin{align*}
&\hspace{-1cm}\textrm{\bf Sudan} \hspace{-2cm}& t &\le\frac{2n\ell-\ell(\ell+1)\deg(G)}{2(\ell+1)}-g -\frac{1}{\ell+1}\\
&\hspace{-1cm}\textrm{\bf Power\ decoding} \hspace{-2cm}& t &\le\frac{2n\ell-\ell(\ell+1)\deg(G)}{2(\ell+1)}-g -\frac{\ell}{\ell+1}\cdot
\end{align*}
First, note that if
\[
g>\ell-1,
\]
then the decoding radius of the $\ell$--PELP algorithm
(\ref{Dec_rad:PELP_AG}) is even larger than Sudan's algorithm decoding
radius. Furthermore, one can see that for algebraic geometry codes,
the power decoding algorithm and the power error locating pairs
algorithm no longer have the same decoding radius, but the second one
is larger. Actually the implementation of the algorithms put in
evidence that power decoding algorithm is actually able to correct
more than what expressed by its decoding radius, and up to the
recoding radius of the $\ell$--PELP algorithm. It is possible to
explain this by considering that in the power decoding algorithm
something changes once we pass to algebraic geometry codes from
Reed--Solomon codes. Indeed, in both cases, the decoding radius comes
as a necessary condition for a vector space to have dimension one. But
for Reed--Solomon codes, this is equivalent to have a necessary
condition for the algorithm to succeed, while for algebraic geometry
codes this is no longer true.

\medskip By the tests we made, it seems that the bound
(\ref{Dec_rad:PELP_AG}) is optimal. Though we should precise that
we run the algorithms with $\deg(F)=t+2g$. Actually experimentally we
have seen that it is possible to run power decoding algorithm with
$\deg(F)=t+g$ and obtain an empirical decoding radius
\[t
  \le\frac{2n\ell-\ell(\ell+1)\deg(G)}{2(\ell+1)}-\frac{\ell}{\ell+1}\cdot\]
which indeed corresponds to the empirical result obtained in \cite{NB15}.

\subsection{Cryptanalytic application}
In the last fourty years, many attempts for instantiating McEliece
encryption scheme \cite{M78} using algebraic codes have been proposed
in the literature. The use of generalised Reed--Solomon codes is known
to be unsecure since Sidelnikov and Shestakov's attack \cite{SS92}
permitting to recover the whole structure of such a code from the very
knowledge of a generator matrix.  Note that actually, a procedure to
recover the structure of a generalised Reed--Solomon code from the
data of a generator matrix was already known by Roth and Seroussi
\cite{RS85}. Sidelnikov--Shestakov attack has been extended to
algebraic geometry codes from curves of genus 1 and 2
\cite{M07,FM08}. For general algebraic geometry codes, an attack has
been given \cite{CMP17} that permits to recover an error correcting
pair or an {\em error correcting array} from the knowledge of a
generating matrix. This attack does not permits to recover the curve,
the divisor and the evaluation points but is enough to break the
system as soon as the decoder corrects at most half the designed
distance.

In a nutshell, this attack of \cite{CMP17} consists in computing some
filtered sequences of codes from the knowledge of a generator matrix
of $\AGcode{\X}{\points}{G}$. Namely, the codes computed are of the
form $\AGcode{\X}{\points}{iP}$ and $\AGcode{\X}{\points}{G-iP}$ for a
given rational point $P$ and for any integer $i$. For $i$ large
enough, the pair
$(\AGcode{\X}{\points}{iP}, \AGcode{\X}{\points}{G-iP})$ yields an
error correcting pair.

Suppose now that McEliece scheme is instantiated with an algebraic geometry
code and whose decryption step requires to correct beyond half the designed
distance by using Sudan's or power decoding algorithm. {\em Stricto sensu},
such a scheme is out of reach by the attack \cite{CMP17}. However,
the very same approach permits to design a power error locating pair. Then,
Algorithm~\ref{algoPELP2} can be run without requiring any further knowledge
on the curve and the divisor. 
This yields an interesting application of
this abstract formulation of decoding beyond half the minimum distance.
Note that no such cryptographic proposal exists in the literature but
\cite{ZZ18} which is unfortunately out of reach of power error locating pairs
since it requires the use of a Guruswami--Sudan like decoder yielding a decoding
radius close to Johnson bound.

 \section{PELP algorithm for cyclic
  codes}\label{sec:cyclic}
In this section, we give an application of the PELP algorithm for some
cyclic codes. In 1994, Duursma and K\"otter showed in \cite{DK94} that
an ECP algorithm can correct up to half the BCH bound and, in
particular cases, also half the Roos bound (see
Theorem~\ref{Roos_bound} for a definition and \cite{R83} for details).  

\medskip First, we recall the main notions and fix some
notation (for more details see \cite{DK94}). Let us consider a field $\mathbb{F}_q$ and an integer
$n$ with $\gcd(n, q)=1$. Given a vector
$\cv=(c_0,\dots,c_{n-1})\in\mathbb{F}_q^n$, we denote by $c(X)$ the
image of $\cv$ by the following linear map:
\[\map{\mathbb{F}_q^n}{\mathbb{F}_q[X]\slash(X^n-1)}
  {(c_0,\dots,c_{n-1})}{\sum_{i=0}^{n-1}c_iX^i.}\]
It is known that cyclic codes of length $n$ over $\mathbb{F}_q$ are in
correspondence with the factors of the polynomial $X^n-1$. In
particular, given $g(X)|X^n-1$ in $\mathbb{F}_q[X]$, the cyclic code
$C_g$ associated to $g$ is
\[C_g\eqdef\{\cv\in\mathbb{F}_q^n\ \ \textrm{such\ that}\ \
  g(X)|c(X)\}.\] In the same way, the roots of $g$ determine in a unique way the code $C_g$. Hence, let us consider an
extension $\mathbb{F}\supseteq\mathbb{F}_q$ such that $\mathbb{F}$
contains the $n$--th roots of unity and let $\gamma$ be a primitive
$n$--th root of unity.
\begin{defn}\label{matrix_cyclic}
  Given $R=\{i_1,\dots, i_m\}\subseteq \{1,\dots, n\}$, we define the
  $m\times n$ matrix
\[M(R)\eqdef
\begin{pmatrix}
1 & \gamma^{i_1} & \cdots & \gamma^{i_1(n-1)}\\
1 & \gamma^{i_2} & \cdots & \gamma^{i_2(n-1)}\\
\vdots & \vdots  &        & \vdots\\
1 & \gamma^{i_m} & \cdots & \gamma^{i_m(n-1)}
\end{pmatrix}.\]
\end{defn}
\medskip\noindent To any subset $R\subseteq\{1,\dots, n\}$, one can
then associate two cyclic codes.
\begin{defn}
$R$ is called \textit{defining set} for the code $C$ if 
\begin{equation}\label{def_set}
C=\{\cv\in \mathbb{F}_q^n\mid M(R)\cv^T=0\}.
\end{equation}
\end{defn}
\medskip\noindent 
\begin{rem}
  One can see that if $C$ is defined as in \eqref{def_set}, then $C$
  is a cyclic code. Indeed, we have $C=C_g$, where
  $g=\lcm\{m_i(x)\mid i\in R\}$ and $m_i$ is the minimal polynomial of
  $\gamma^i$ on $\mathbb{F}_q$. Note that different defining sets can
  define the same cyclic code $C$. By applying several times Frobenius
  morphism to the set $\{\gamma^i\mid i\in R\}$, one can find the
  maximal defining set, also called \textit{complete}. In this paper
  we will treat a general situation, where a defining set will not
  necessarily be complete.
\end{rem}
\begin{rem}
Note that, if $R$ is a defining set for a code $C$,
then $C=\tilde{C}\cap\mathbb{F}_q^n$, where
$\tilde{C}\subseteq\mathbb{F}^n$ is a cyclic code with parity check matrix
$M(R)$. If we denote by $\dd_R$ the minimum distance of the code
$\tilde{C}$, we get $\dd(C)\ge\dd_R$.
\end{rem}
\begin{defn}\label{def_gen_set}
$R$ is called \textit{generating set} for the code $C$ if 
\begin{equation}\label{gen_set}
C=\{\av M(R) ~|~\av\in\mathbb{F}^m\}.
\end{equation}
\end{defn}
\medskip\noindent We stress that if $R$ is a generating set for a code
$C$, then $C$ is a code with coefficients in the larger alphabet
$\mathbb{F}$ and has generating matrix $M(R)$. In particular,
$\dim_{\F}(C)=|R|$.
\begin{rem}
  Note that a code $C$ as in \eqref{gen_set} is a cyclic code. Indeed
  $C$ is the dual code of the cyclic code $D\subseteq\mathbb{F}^n$
  with defining set $R$ and it is known that the dual of a cyclic code
  is cyclic itself.
\end{rem}

\subsection{Roos bound} 
There are cases where it is possible to bound the minimum distance of
a cyclic code. Apart from the BCH bound, another and more general
bound has been given by Roos (\cite{R83}).
\begin{defn}\label{S+R}
  Given $R\subseteq\{1,\dots, n\}$, denote by $\overline{R}$ the smallest
  set made of consecutive indices modulo $n$ that contains
  $R$. Moreover, if $S$ is another subset of $\{1,\dots, n\}$, we can
  define the sum set
  \[S+R\eqdef\{s+r\mod n\mid s\in S, r\in R\}.\] Finally, given
  $a < n$, we define the set $aR\eqdef\{ar\mod n\mid r\in R\}$.
\end{defn}
It is possible to relate the star product of two cyclic codes to the
sum of their generating sets.
\begin{prop}\label{prop:S+R}
  Let $A$, $B$ and $W$ be three cyclic codes with generating sets
  respectively $S$, $R$ and $S+R$. Then,
\[A\ast B= W.\]
\end{prop}
\begin{proof}
First, note that for any $j\in S+R$ we get by Definition \ref{S+R}
\[j=s+r \mod n\]
for some $s\in S$ and $r\in R$. Hence,
\begin{equation}\label{easy_equality}
(1, \gamma^j, \dots, \gamma^{j(n-1)})=(1, \gamma^s, \dots, \gamma^{s(n-1)})\ast (1, \gamma^r, \dots, \gamma^{r(n-1)}).
\end{equation}
Now, it is easy to see that the set of generators of $A\ast B$
\[G \eqdef \{(1, \gamma^{s+r},\dots, \gamma^{(s+r)(n-1)} )\mid s\in S, r\in
  R\}\] is equal to the set composed by the rows of the matrix
$M(S+R)$ (see Definition~\ref{matrix_cyclic}). Since, by
Definition~\ref{def_gen_set}, this is a generator matrix for the code
$W$, we get that $G$ is a set of generators for both $A\ast B$ and
$W$, hence $A\ast B=W$.
\end{proof}

\begin{thm}[Roos bound]\label{Roos_bound}
  Let $R, S\subseteq\{1,\dots, n\}$ such that
  $|\overline{S}|\le |S|+\dd_R-2$. Then,
\[\dd_{R+S}\ge |S|+\dd_R-1.\]
\end{thm}
\begin{proof}
See \cite{R83}.
\end{proof}
\begin{rem}
  In the hypothesis of Theorem \ref{Roos_bound}, if $C$ is the cyclic
  code with defining set $R+S$, since $\dd(C)\ge\dd_{R+S}$, then
  $\dd(C)\ge |S|+\dd_R-1$ as well.
\end{rem}
\begin{rem}
  One can note that in the same hypothesis of Theorem
  \ref{Roos_bound}, the proof given in \cite{R83} can be adapted to
  prove that
\[\dd_{aR+bS}\ge |S|+\dd_R-1\]
for any $a, b\le n$ with $\gcd (a,n)=\gcd (b,n)=1$.

\end{rem}

\subsection{$\ell$--PELP algorithm and Roos bound}
We now focus on cyclic codes with defining set $R+S$ with $R$ and $S$
satisfying the hypothesis of Roos bound
(Theorem~\ref{Roos_bound}). Actually, we will work with the code in
$\mathbb{F}^n$ for the sake of simplicity.

\begin{thm}\label{thm:PELP_Cyclic}
Let $a, b\le n$ with $\gcd (a, n)= \gcd (b, n)=1$ and let
  $A,B\subseteq\mathbb{F}^n$ be cyclic codes with generating sets
  respectively $aS$ and $bR$, where
  \[|\overline{S}|\le |S|+\dd_R-2,\qquad |S|>t,\qquad \dd_S>t.\] Let
  $\tilde{C}\subseteq\mathbb{F}^n$ be the cylic code with parity check
  matrix $M(aS+bR)$ and $k=\dim(\tilde{C})$. Let us suppose that
  \begin{enumerate}[(i)]
  \item for any $i \in \{1, \dots, \ell-1\}$ we have
    $B^{\perp}\ast \tilde{C}^i\varsubsetneq \mathbb{F}_q^n$;
  \item\label{item:subcodes_of_B} any nonzero cyclic subcode of $B$ is
    non degenerated.
  \end{enumerate}
 Then $(A, B)$
  is an $\ell$--power $t$--error locating pair for the code
  $\tilde{C}$ with
  \begin{equation}\label{dec_radius_ell_cyclic}
  t\le \ell n-\Big[\frac{\ell(\ell+1)}{2}(k-1) + \ell(|S|+\delta)+\sum_{i=1}^{\ell-1}\gamma_i\Big],
  \end{equation}
  where $\delta$ and $\gamma_1,\dots, \gamma_{\ell-1}$ fulfill 
  \begin{eqnarray}
  n-k &=&|S|+|R|-1 +\delta \nonumber\\
  \dim(B) &=&\dim((B^{\perp}\ast \tilde{C}^i)^{\perp})+i\dim(\tilde{C})-i+\gamma_i\ \ \ \forall i=1,\dots, \ell-1.\label{cond_cyclic}
  \end{eqnarray}
  \end{thm}

\begin{rem}
  Note that if $(A,B)$ is an $\ell$--power $t$--error locating pair
  for $\tilde{C}$, then it is an $\ell$--power $t$--error locating pair
  for the cyclic code $C$ with defining set $aS+bR$ as well. Actually
  it is a standard procedure for cyclic codes (see for instance
  \cite{DK94}). In particular, that is why if $\tilde{C}$ is a
  Reed--Solomon code, the optimised choice of PELP for $\tilde{C}$
  with $|S|=t+1$, will give the decoding radius found in
  \S~\ref{sec:PELP_RS}.
\end{rem}

\begin{rem}\label{rem:reformulation}
  Condition~(\ref{item:subcodes_of_B}) on $B$ can be reformulated as follows:
  for any non empty subset $U$ of $R$, there does not exist $i \in \Z/n\Z$
  such that $U + i \equiv U \mod n$.
\end{rem}

Before proving Theorem~\ref{thm:PELP_Cyclic}, we need the two following
lemmas on the notion of {\em degenerated} codes (see \S~\ref{sss:Kneser}).

\begin{lem}\label{lem:prod_degenerated}
  Let $C \subseteq \Fq^n$ be a degenerated code. Then for any code $D \subseteq
  \Fq^n$, the code $C \ast D$ is degenerated too.
\end{lem}

\begin{proof}
  It suffices to observe that
  $\textrm{Stab}(C) \subseteq \textrm{Stab}(C \ast D)$.
\end{proof}

\begin{lem}\label{lem:dual_degenerated}
  A code $C \subseteq \Fq^n$ is degenerated if and only if $C^{\perp}$
  is degenerated.
\end{lem}

\begin{proof}
  Using the adjunction property (\ref{eq:adjunction}) of the star product,
  one proves that $\textrm{Stab}(C) = \textrm{Stab}(C^\perp)$.
\end{proof}

\begin{proof}[Proof of Theorem~\ref{thm:PELP_Cyclic}]
  We treat the case $a=b=1$, the general case being an easy
  generalisation. We have by hypothesis $\dim(A)=|S|>t$. Next,
  from Proposition~\ref{prop:S+R}, 
  $A\ast B=\tilde{C}^{\perp}$. Furthermore, we have
  $\dd(A^{\perp})=\dd_S>t$. Hence, properties
  ($\ell$--PELP\ref{item:pelp1g}), ($\ell$--PELP\ref{item:pelp2g}) and
  ($\ell$--PELP\ref{item:pelp3g}) are satisfied.

  Now, let us focus on
  property ($\ell$--PELP\ref{item:pelp4g}). We have that $A$ is
  contained in the code with generating set $\overline{S}$, whose
  distance is $n-|\overline{S}|+1$ (note that it is a generalised
  Reed--Solomon code). Hence, we get $\dd(A)\ge n-|\overline{S}|+1$,
  which, together with Roos bound, gives
  \[\dd(A)+\dd(\tilde{C})\ge n-|\overline{S}|+|S|+\dd_R\ge n+2>n.\]
  
  In order to check Property ($\ell$--PELP\ref{item:pelp5g}),
  we first consider the case $\ell = 2$. Set
  $W\eqdef~(B^{\perp}\ast~\tilde{C})^{\perp}$. Then,
  \begin{equation}\label{eq:last_inclusion}
    W\perp B^{\perp}\ast \tilde{C}\iff W\ast \tilde{C}\subseteq B.
  \end{equation}
  From Condition~(\ref{item:subcodes_of_B}) on $B$, the code
  $W \ast \tilde C$ is non degenerated.
  This last observation, together with inclusion
  (\ref{eq:last_inclusion}) and
  Corollary~\ref{cor:CauchyDavenport} yield
  \[\dim(W)+\dim(\tilde{C})-1+\gamma_1= \dim(B)\]
  for some nonnegative integer $\gamma_1$.  Next, since
  $W \ast \tilde C$ is non degenerated, from
  Lemmas~\ref{lem:prod_degenerated} and~\ref{lem:dual_degenerated},
  the code $W^\perp = B^\perp \ast C$ is non
  degenerated too.  Thus, we get 
  \begin{eqnarray}
    (2\textrm{--PELP}\ref{item:pelp5}):\ \dim(B)+\dim(B^{\perp}\ast \tilde{C})^{\perp} \ge t&\iff &2\dim(B)-\dim(\tilde{C})+1-\gamma_1\ge t \nonumber\\
                                                                                                 &\iff &2|R|-k+1-\gamma_1\ge t.\label{dec_radius_cyclic2}
  \end{eqnarray}
  Now, since $A\ast B = \tilde{C}^{\perp}$, using
  Corollary~\ref{cor:CauchyDavenport} again, we know that there exists
  $\delta\ge 0$ such that
  \begin{equation}\label{dec_radius_cyclic1}
    |S|+|R|-1 +\delta= n-k.
  \end{equation}
  Hence, by (\ref{dec_radius_cyclic2}) and (\ref{dec_radius_cyclic1}),
  property ($2$--PELP\ref{item:pelp5}) is equivalent to
  \[t\le 2n-3k+3-2\delta-\gamma_1 -2|S|.\] It is now easy to
  generalise the proof for $\ell>2$. Indeed if we consider $i>2$ and
  $Z\eqdef(B^{\perp}\ast \tilde{C}^i)^{\perp}$, we have as before
  \[Z\ast \tilde{C}^i\subseteq B.\] From
  Condition~(\ref{item:subcodes_of_B}) together with
  Lemma~\ref{lem:prod_degenerated}, we deduce that $\tilde C^i$ is non
  degenerated. Then, by applying Corollary~\ref{cor:CauchyDavenport}
  iteratively and thanks to Condition~(\ref{item:subcodes_of_B})
  again, we know that there exist two nonnegative integers $\gamma'$
  and $\gamma''$ such that
  \begin{align*}
    \dim(\tilde{C}^i)&=i\dim(\tilde{C})-i+1+\gamma'\\
    \dim(B) &=\dim(Z)+\dim(\tilde{C}^i)-1+\gamma''.
  \end{align*}
  Note that these two equations give \eqref{cond_cyclic} with
  $\gamma_i\eqdef \gamma'+\gamma''$. Finally, by combining
  \eqref{cond_cyclic} and \eqref{dec_radius_cyclic1}, we get that
  Property ($\ell$--PELP\ref{item:pelp5g}) for $\ell>2$ is equivalent
  to the bound in \eqref{dec_radius_ell_cyclic}.
\end{proof}
\begin{rem}
  Note that $\delta$ and the $\gamma_i$'s do not depend only on the choice of
  $R$ and $S$ but also on the parameters $a, b$. Hence, in particular,
  the decoding radius depends as well on $a,b$.
\end{rem}

\subsection{Comparison with Roos bound}
We now would like to compare the obtained decoding radius to Roos
bound. To do so, we consider a particular case of cyclic code. Let
$R,S\subseteq\{1,\dots, n\}$ such that $|R|=|\overline{R}|=r$, $|S|>t$
and $|\overline{S}|\le|S|+\dd_R-2$. Let us denote by $t_{A,B}$ the
decoding radius (\ref{dec_radius_ell_cyclic}) for $\ell=2$ and by
$\dd_{Roos}$ the amount $|S|+\dd_R-1$. By the equality
$|S|+|R|-1+\delta=n-k$, that is $|S|+r-1+\delta=n-k$, we get
\begin{equation}\label{comp:Roos_dec_radius}
  t_{A,B}\ge \frac{\dd_{Roos}-1}{2}\iff k\le\frac{3n+6-3\delta-2\gamma_1 - 4|S|}{5}\cdot
\end{equation}

\begin{rem}
  Observe that, if $\gcd(a,n)=\gcd(b,n)=1$, Roos bound holds even for $aS$ and
  $bR$. Hence (\ref{comp:Roos_dec_radius}) gives an useful information
  about the behaviour of $t_{A,B}$ for any $a,b$ $\in\mathbb{N}$ with
  $(a,n)=(b,n)=1$.
\end{rem} 

\medskip
\noindent
Equivalence (\ref{comp:Roos_dec_radius}) gives a good information
about the parameters to have in order to cross half the Roos
bound. Indeed by some tests we made it has been possible to see that
the decoding radius $t_{A, B}$ is achieved really often.

\begin{exmp}
  We now give an easy example of a $2$--PELP algorithm's application
  on a cyclic code (which is not BCH) where
  $t_{A,B}>\frac{\dd_{Roos}-1}{2}$. Let us consider $n=51$, $q=5$ and
  the sets $S=\{0,\dots, 24\}\cup\{30\}$,
  $R=\{0,\dots,13\}\cup\{19\}$.
  According to Remark~\ref{rem:reformulation}, one can check that
  Condition~(\ref{item:subcodes_of_B}) of Theorem~\ref{thm:PELP_Cyclic}
  is satisfied by $B$.
  We now consider the cyclic code
  $\tilde{C}$ with parity check matrix $M(S+R)$. Since we are in the
  hypothesis of Roos bound and $\dd_R=15$, we obtain
  $\frac{\dd_{Roos}-1}{2}=19$, while $t_{A,B}=23$. In fact, the true
  minimum distance of $\tilde C$ can be computed to be 45. Hence we
  get
\[t_{A,B}=23>\frac{\dd(\tilde{C})-1}{2}=22.\]
\end{exmp}

  \section*{Conclusion}
 We proposed a unified framework for a decoder that can correct beyond
  half the minimum distance. Exactly as error correcting pairs can be
 regarded as an abstraction of Welch--Berlekamp algorithm, our
 approach called power error locating pairs is an abstraction of power
 decoding for Reed--Solomon and algebraic geometry codes.  This
 algorithm applies to any code equipped to a power error locating pair
 structure such as some cyclic codes for instance.  In addition our
 results turn out to have interesting consequences on code based
 cryptography since we proved that a McEliece like system using
 algebraic geometry codes with a secret decoder correcting up to
 Sudan's radius is unsecure.

 On the other hand, our algorithm does not decode Reed--Solomon or
 algebraic geometry codes up to the Johnson radius.  In this
 direction, finding an abstraction of Rosenkilde's extension of power
 decoding \cite{N15} would represent an interesting challenge. Such a
 result would for instance yield an attack to any cryptosystem like the
 one introduced in \cite{ZZ18}.

\section*{Acknowledgements}
The authors express their gratitude to the anonymous referees
for their careful work and their many relevant comments permitting
a significant improvement of this article.
This work was supported by French {\em Agence Nationale de la Recherche}
{\sc Manta} : ANR-15-CE39-0013.

\nocite{*}
\bibliographystyle{alpha}

\begin{thebibliography}{{McE}03b}

\bibitem[Ber68]{B68}
Elwyn~R. Berlekamp.
\newblock {\em Algebraic coding theory}.
\newblock McGraw-Hill Book Co., New York-Toronto, Ont.-London, 1968.

\bibitem[Ber15]{B15}
Elwyn~R. Berlekamp.
\newblock {\em Algebraic coding theory}.
\newblock World Scientific Publishing Co. Pte. Ltd., Hackensack, NJ, revised
  edition, 2015.

\bibitem[BH08]{BH08}
Peter Beelen and Tom H{\o}holdt.
\newblock The decoding of algebraic geometry codes.
\newblock In {\em Advances in algebraic geometry codes}, volume~5 of {\em Ser.
  Coding Theory Cryptol.}, pages 49--98. World Sci. Publ., Hackensack, NJ,
  2008.

\bibitem[BL17]{BL17}
Vincent Beck and C{\'e}dric Lecouvey.
\newblock Additive combinatorics methods in associative algebras.
\newblock {\em Confluentes Math.}, 9(1):3--27, 2017.

\bibitem[CMCP17]{CMP17}
Alain Couvreur, Irene M{\'a}rquez-Corbella, and Ruud Pellikaan.
\newblock Cryptanalysis of {M}c{E}liece {C}ryptosystem {B}ased on {A}lgebraic
  {G}eometry {C}odes and {T}heir {S}ubcodes.
\newblock {\em IEEE Trans. Inform. Theory}, 63(8):5404--5418, Aug 2017.

\bibitem[COT17]{COT17}
Alain Couvreur, Ayoub Otmani, and Jean-Pierre Tillich.
\newblock Polynomial time attack on wild {M}c{E}liece over quadratic
  extensions.
\newblock {\em IEEE Trans. Inform. Theory}, 63(1):404--427, Jan 2017.

\bibitem[DK94]{DK94}
Iwan~M. {Duursma} and Ralf {K{\"o}tter}.
\newblock Error-locating pairs for cyclic codes.
\newblock {\em IEEE Trans. Inform. Theory}, 40(4):1108--1121, July 1994.

\bibitem[Duu93]{D93}
Iwan~M. Duursma.
\newblock {\em Decoding codes from curves and cyclic codes}.
\newblock PhD thesis, Technische Universiteit Eindhoven, 1993.

\bibitem[FM08]{FM08}
C{\'e}dric Faure and Lorenz Minder.
\newblock Cryptanalysis of the {McEliece} cryptosystem over hyperelliptic
  curves.
\newblock In {\em Proceedings of the eleventh International Workshop on
  Algebraic and Combinatorial Coding Theory}, pages 99--107, Pamporovo,
  Bulgaria, June 2008.

\bibitem[GS92]{GS92}
Peter Gemmell and Madhu Sudan.
\newblock Highly resilient correctors for polynomials.
\newblock {\em Inform. Process. Lett.}, 43(4):169 -- 174, 1992.

\bibitem[GS99]{GS99}
Venkatesan Guruswami and Madhu Sudan.
\newblock Improved decoding of {R}eed--{S}olomon and {A}lgebraic--{G}eometry
  codes.
\newblock {\em IEEE Trans. Inform. Theory}, 45(6):1757--1767, 1999.

\bibitem[GS00]{GSA00}
Shuhong Gao and M.~Amin Shokrollahi.
\newblock Computing {R}oots of {P}olynomials over {F}unction {F}ields of
  {C}urves.
\newblock In David Joyner, editor, {\em Coding Theory and Cryptography}, pages
  214--228, Berlin, Heidelberg, 2000. Springer Berlin Heidelberg.

\bibitem[GV05]{GV05}
Venkatesan {Guruswami} and Alexander {Vardy}.
\newblock Maximum--likelihood decoding of {R}eed--{S}olomon codes is {NP}-hard.
\newblock {\em IEEE Trans. Inform. Theory}, 51(7):2249--2256, July 2005.

\bibitem[HP95]{HP95}
Tom {H{\o}holdt} and Ruud {Pellikaan}.
\newblock On the decoding of algebraic-geometric codes.
\newblock {\em IEEE Trans. Inform. Theory}, 41(6):1589--1614, Nov 1995.

\bibitem[JH04]{JH04}
J{\o}rn Justesen and Tom H{\o}holdt.
\newblock {\em A Course in Error-Correcting Codes}.
\newblock European Mathematical Society Publishing House, first edition, 2004.

\bibitem[JLJ{\etalchar{+}}89]{JLJHH89}
J{\o}rn {Justesen}, Knud~J. {Larsen}, Helge~E. {Jensen}, Allan {Havemose}, and
  Tom {Hoholdt}.
\newblock Construction and decoding of a class of algebraic geometry codes.
\newblock {\em IEEE Trans. Inform. Theory}, 35(4):811--821, July 1989.

\bibitem[{Joh}62]{J62}
Selmer~M. {Johnson}.
\newblock A new upper bound for error-correcting codes.
\newblock {\em IRE Trans. Inform. Theory}, 8(3):203--207, April 1962.

\bibitem[K{\"o}t92]{K92}
Ralf K{\"o}tter.
\newblock A unified description of an error locating procedure for linear
  codes.
\newblock In {\em Proceedings Algebraic and Combinatorial Coding Theory III},
  pages 113--117. Hermes, 1992.

\bibitem[KS80]{KS80}
Wolfgang Knapp and Peter Schmid.
\newblock Codes with prescribed automorphism group.
\newblock {\em J. Algebra}, 67(2):415--435, 1980.

\bibitem[McE78]{M78}
Robert~J. McEliece.
\newblock {\em A Public-Key System Based on Algebraic Coding Theory}, pages
  114--116.
\newblock Jet Propulsion Lab, 1978.
\newblock DSN Progress Report 44.

\bibitem[McE03a]{M03}
Robert~J. McEliece.
\newblock On the {A}verage {L}ist {S}ize for the {G}uruswami-{S}udan decoder.
\newblock In {\em 7th Inernational Symposium on Communications Theory and
  Applications (ISCTA)}, 2003.

\bibitem[{McE}03b]{M031}
Robert~J. {McEliece}.
\newblock {The Guruswami--Sudan Decoding Algorithm for Reed--Solomon Codes}.
\newblock {\em Interplanetary Network Progress Report}, 153:1--60, January
  2003.

\bibitem[Min07]{M07}
Lorenz Minder.
\newblock {\em Cryptography based on error correcting codes}.
\newblock PhD thesis, Ecole Polytechnique F{\'e}d{\'e}rale de Lausanne, 2007.

\bibitem[Mum70]{M70}
David Mumford.
\newblock Varieties defined by quadratic equations.
\newblock In {\em Questions on algebraic varieties, C.I.M.E., III Ciclo,
  Varenna, 1969}, pages 29--100. Edizioni Cremonese, Rome, 1970.

\bibitem[MZ15]{MZ15}
Diego Mirandola and Gilles Z{\'e}mor.
\newblock Critical pairs for the product {S}ingleton bound.
\newblock {\em IEEE Trans. Inform. Theory}, 61(9):4928--4937, 2015.

\bibitem[NB15]{NB15}
Johan S.~R. {Nielsen} and Peter {Beelen}.
\newblock Sub--quadratic decoding of one-point {H}ermitian codes.
\newblock {\em IEEE Trans. Inform. Theory}, 61(6):3225--3240, June 2015.

\bibitem[Pel88]{P88}
Ruud Pellikaan.
\newblock On decoding linear codes by {E}rror {C}orrecting {P}airs.
\newblock Preprint Technical University Eindhoven, 1988.

\bibitem[Pel92]{P92}
Ruud Pellikaan.
\newblock On decoding by error location and dependent sets of error positions.
\newblock {\em Discrete Math.}, 106--107:369--381, 1992.

\bibitem[PRB19]{PRB19}
Sven Puchinger, Johan Rosenkilde, and Irene Bouw.
\newblock Improved power decoding of interleaved one--point {H}ermitian codes.
\newblock {\em Des. Codes Cryptogr.}, 87:589--607, 2019.

\bibitem[PV05]{PV05}
Farzad Parvaresh and Alexander Vardy.
\newblock Correcting errors beyond the {G}uruswami-{S}udan radius in polynomial
  time.
\newblock In {\em Foundations of Computer Science, 2005. FOCS 2005. 46th Annual
  IEEE Symposium on}, pages 285--294, 2005.

\bibitem[RnN15]{R15}
Johan Rosenkilde~(né Nielsen).
\newblock Power {D}ecoding of {R}eed--{S}olomon {C}odes {R}evisited.
\newblock In {\em Coding Theory and Applications}, pages 297--305, Cham, 2015.
  Springer International Publishing.

\bibitem[RnN18]{N15}
Johan Rosenkilde~(né Nielsen).
\newblock Power decoding {R}eed-{S}olomon codes up to the {J}ohnson radius.
\newblock {\em Adv. in Math. of Comm.}, 12:81--106, 2018.

\bibitem[Roo83]{R83}
Cornelis Roos.
\newblock A new lower bound for the minimum distance of a cyclic code.
\newblock {\em IEEE Trans. Inform. Theory}, 29:330 -- 332, June 1983.

\bibitem[Rot06]{R06}
Ron~M. Roth.
\newblock {\em Introduction to coding theory}.
\newblock Cambridge university press, 2006.

\bibitem[RS85]{RS85}
R.~M. {Roth} and G.~{Seroussi}.
\newblock On generator matrices of {MDS} codes ({C}orresp.).
\newblock {\em IEEE Trans. Inform. Theory}, 31(6):826--830, 1985.

\bibitem[RW14]{RW14}
Atri Rudra and Mary Wootters.
\newblock Every list-decodable code for high noise has abundant near-optimal
  rate puncturings.
\newblock In {\em Proceedings of the Forty-Sixth Annual ACM Symposium on Theory
  of Computing}, STOC ’14, pages 764--773, New York, NY, USA, 2014.
  Association for Computing Machinery.

\bibitem[SS92]{SS92}
Vladimir~Michilovich Sidelnikov and S.O. Shestakov.
\newblock On the insecurity of cryptosystems based on generalized
  {Reed-Solomon} codes.
\newblock {\em Discrete Math. Appl.}, 1(4):439--444, 1992.

\bibitem[SSB09]{SSB09}
Georg {Schmidt}, Vladimir~R. {Sidorenko}, and Martin {Bossert}.
\newblock Collaborative decoding of interleaved {R}eed--{S}olomon codes and
  concatenated code designs.
\newblock {\em IEEE Trans. Inform. Theory}, 55(7):2991--3012, July 2009.

\bibitem[SSB10]{SSB10}
Georg Schmidt, Vladimir~R. Sidorenko, and Martin Bossert.
\newblock Syndrome {D}ecoding of {R}eed–{S}olomon {C}odes {B}eyond {H}alf the
  {M}inimum {D}istance {B}ased on {S}hift-{R}egister {S}ynthesis.
\newblock {\em IEEE Trans. Inform. Theory}, 56(10):5245--5252, Oct 2010.

\bibitem[Sti09]{S09}
Henning Stichtenoth.
\newblock {\em Algebraic function fields and codes}, volume 254 of {\em
  Graduate Texts in Mathematics}.
\newblock Springer-Verlag, Berlin, second edition, 2009.

\bibitem[Sud97]{S97}
Madhu Sudan.
\newblock Decoding of {Reed--Solomon} {C}odes beyond the {E}rror-{C}orrection
  {B}ound.
\newblock {\em J. Complexity}, 13(1):180--193, 1997.

\bibitem[SV90]{SV90}
Aleksei~N. {Skorobogatov} and Sergei~G. {Vl\u{a}du\c{t}}.
\newblock On the decoding of algebraic-geometric codes.
\newblock {\em IEEE Trans. Inform. Theory}, 36(5):1051--1060, Sep. 1990.

\bibitem[SW99]{SW99}
M.~Amin {Shokrollahi} and Hal {Wasserman}.
\newblock List decoding of algebraic-geometric codes.
\newblock {\em IEEE Trans. Inform. Theory}, 45(2):432--437, March 1999.

\bibitem[TV06]{TV06}
T.~Tao and V.~H. Vu.
\newblock {\em Additive combinatorics}, volume 105 of {\em Cambridge studies in
  advanced mathematics}.
\newblock Cambridge University Press, 2006.

\bibitem[TVN07]{TVN07}
Michael Tsfasman, Sergei Vl{\u{a}}du{\c{t}}, and Dmitry Nogin.
\newblock {\em Algebraic geometric codes: basic notions}, volume 139 of {\em
  Mathematical Surveys and Monographs}.
\newblock American Mathematical Society, Providence, RI, 2007.

\bibitem[WB83]{WB83}
Lloyd~R. Welch and Elwyn~R. Berlekamp.
\newblock Error correction for algebraic block codes, 1983.
\newblock US patent number 4,633,470.

\bibitem[ZZ18]{ZZ18}
Fangguo Zhang and Zhuoran Zhang.
\newblock Code-based cryptosystem from quasi-cyclic elliptic codes.
\newblock Cryptology ePrint Archive, Report 2018/1182, 2018.
\newblock \url{https://eprint.iacr.org/2018/1182}.

\end{thebibliography}
\newcommand{\etalchar}[1]{$^{#1}$}

\appendix
\section{Power decoding for algebraic geometry
  codes}\label{sec:appendix}
We show how
the power decoding algorithm adapts
to arbitrary algebraic geometry codes. Let $C=C_L(\X, \mathcal{P}, G)$ and
$\textbf{y}=\textbf{c}+\textbf{e}\in\mathbb{F}_q^n$ the word we want
to correct, where $\textbf{c}\in C$. We have then
$$c=\ev_{\mathcal{P}}(f)\text{ with }f\in L(G).$$ 
Furthermore, as in the previous sections, we suppose that
$\w(\textbf{e})=t$ and denote the support of $\textbf{e}$ by ${I_{\error}}$.
We keep the same idea we used in the version of the algorithm for
Reed--Solomon codes. Indeed, let us suppose to have 
$\Lambda\in\mathbb{F}_q(\X)$ such that $\Lambda(P_i)=0$ for all
$i\in {I_{\error}}$. Then, given $\ell\in\mathbb{N}$ we get
\begin{equation}
\Lambda(P_i)y_i^j=\Lambda(P_i)f^j(P_i)\ \ \ \ \forall\ i=1,\dots, n,\ \ j=1,\dots\ell.
\end{equation}
We would like then to find $\Lambda$ as above. It is easy to see that
such a $\Lambda$ has to be searched in $L(F)$ for a certain $F$
such that $\deg(F)\ge t+g$. 
(we will give a better bound for that
soon).
It is possible to see $(\Lambda, f)$ as a solution of
\begin{equation}\label{poweqnonlin}
\lambda(P_i)y_i^j=\lambda(P_i)\phi^j(P_i)\ \ \ \ \forall\ i=1,\dots, n,\ \ j=1,\dots\ell,
\end{equation}
that is, a system of $n\ell$ equations whose unknowns are the
coordinates of $\lambda$ and $\phi$ in the basis of respectively
$L(F)$ and $L(G)$. System (\ref{poweqnonlin}) is not
linear in the unknowns though, hence we linearise it by considering a
new function $\nu_j\coloneqq\lambda \phi^j$ for any equation.
For all
$j\in\{1, \dots, \ell\}$, we get
$$\nu_j\in L(F)L(jG)\subseteq L(F+jG).$$ 
In order to use Theorem~\ref{prop:star_prod_AGcodes}, let us fix $\deg(F)= t+2g$ and suppose $\deg(G)\ge 2g+1$. We get then the following problem.
\begin{kprob}
  Given $\textbf{y}\in\mathbb{F}_q^n$ and $t\in\mathbb{N}$, look for
  $\lambda,\nu_1, \dots, \nu_{\ell}\in\mathbb{F}_q^n(\X)$ such that
\begin{itemize}
\item $\lambda\in L(F)$ with $\deg(F)= t+2g$;
\item $\nu_j\in L(F+jG)$ for all $j=1,\dots,\ell$;
\item $\lambda(x_i)y_i=\nu_j(x_i)$ for all $i=1,\dots, n$ and
  $j=1,\dots, \ell$.
\end{itemize}
\end{kprob} 

\medskip\noindent Therefore, even this case, the power decoding algorithm
consists in solving a linear system and we will just consider a
nonzero solution.
\paragraph{Decoding Radius.}
As in the case of Reed--Solomon codes, we would like to have a solution
space of dimension one. A necessary condition for that, is
\begin{equation}\label{PowDecAGcond}
\# unknowns\le \# equations+1.
\end{equation}
The number of equations is $n\ell$. For the number of unknowns,
we need to know the dimension of the spaces $L(F+jG)$ for all
$j=1,\dots, \ell$. The bounds we have set in the hypothesis give
\begin{equation}\nonumber
\dim(L(F+jG))=t+g+j\deg(G)+1.
\end{equation}
Hence by condition (\ref{PowDecAGcond}) we get the following decoding radius
\begin{equation}\label{decradPowDecAG}
t\le\frac{2n\ell-\ell(\ell+1)\deg(G)}{2(\ell+1)}-g -\frac{\ell}{\ell+1}\cdot
\end{equation}
\begin{rem}
  This bound is not a sufficient condition to have a solution, but it
  is not even a necessary condition. In fact, as for the power decoding algorithm for Reed--Solomon codes, we could find a good
  solution even for a larger value of $t$ and on the other hand the
  algorithm could fail even if $t$ fulfills condition
  (\ref{decradPowDecAG}).
\end{rem}

\end{document}